\documentclass[onecolumn,notitlepage,superscriptaddress]{revtex4-1}
\usepackage{dcolumn}
\usepackage{bm}

\usepackage[normalem]{ulem}

\usepackage{colonequals}
\usepackage[colorlinks]{hyperref}
\usepackage{hypcap}
\usepackage{enumitem}
\usepackage{complexity}

\usepackage{amsmath}
\usepackage{amsfonts,amssymb,amsthm, braket}
\usepackage{graphicx}   
\usepackage{amsbsy} 
\usepackage{accents}
\usepackage[bold]{hhtensor}  

\usepackage[bold]{hhtensor}
\usepackage{braket}
\usepackage{pgfplots}
\usetikzlibrary{calc}
\usepackage{pgfplotstable}
\usepackage{colortbl}
\usepackage{booktabs}
\usepackage{pgfplotstable}

\setlist[enumerate]{itemsep=0mm}

\newcommand{\ketbra}[2]{\left|#1\right\rangle\!\left\langle #2\right|}
\newcommand{\braketb}[2]{\left\langle #1 | #2 \right\rangle}

\newcommand{\suppress}[1]{}

\def\squareforqed{\hbox{\rlap{$\sqcap$}$\sqcup$}}
\def\qed{\ifmmode\squareforqed\else{\unskip\nobreak\hfil
\penalty50\hskip1em\null\nobreak\hfil\squareforqed
\parfillskip=0pt\finalhyphendemerits=0\endgraf}\fi}

\newtheorem{theorem}{Theorem}
\newtheorem{lemma}[theorem]{Lemma}

\newtheorem{corollary}[theorem]{Corollary}
\newenvironment{proofof}[1]{\begin{proof}[Proof of~#1.]}{\end{proof}}

\newcommand{\eq}[1]{\hyperref[eq:#1]{(\ref*{eq:#1})}}
\renewcommand{\sec}[1]{\hyperref[sec:#1]{Section~\ref*{sec:#1}}}
\newcommand{\app}[1]{\hyperref[app:#1]{Appendix~\ref*{app:#1}}}
\newcommand{\fig}[1]{\hyperref[fig:#1]{Figure~\ref*{fig:#1}}}
\newcommand{\thm}[1]{\hyperref[thm:#1]{Theorem~\ref*{thm:#1}}}
\newcommand{\lem}[1]{\hyperref[lem:#1]{Lemma~\ref*{lem:#1}}}
\newcommand{\cor}[1]{\hyperref[cor:#1]{Corollary~\ref*{cor:#1}}}
\newcommand{\defn}[1]{\hyperref[def:#1]{Definition~\ref*{def:#1}}}
\newcommand{\alg}[1]{\hyperref[alg:#1]{Algorithm~\ref*{alg:#1}}}
\newcommand{\tab}[1]{\hyperref[tab:#1]{Algorithm~\ref*{tab:#1}}}

\input{Qcircuit}

\begin{document}

\title{Quantum  Algorithms for Nearest-Neighbor Methods for Supervised and Unsupervised Learning}
\author{Nathan Wiebe$^\dagger$}
\author{Ashish Kapoor$^*$}
\author{Krysta M.~Svore$^\dagger$}

\affiliation{$^\dagger$Quantum Architectures and Computation Group, Microsoft Research, Redmond, WA (USA)\\
$^*$Adaptive Systems and Interaction Group, Microsoft Research, Redmond, WA (USA)}

\begin{abstract}
We present quantum algorithms for performing nearest-neighbor learning and $k$--means clustering.
At the core of our algorithms are fast and coherent quantum methods for computing the Euclidean distance both directly and via the inner product which we couple with
methods for performing amplitude estimation that do not require measurement.
We prove upper bounds on the number of queries to the input data required to compute such distances and find the nearest vector to a given test example.
In the worst case, our quantum algorithms lead to polynomial reductions in query complexity relative to Monte Carlo algorithms.
We also study the performance of our quantum nearest-neighbor algorithms on several real-world binary classification tasks and find that the classification accuracy is competitive with classical methods.  
\end{abstract}

\maketitle

Quantum speedups have long been known for problems such as factoring, quantum simulation and optimization.  Only recently have quantum algorithms begun to emerge that promise quantum advantages for solving problems in data processing, machine learning and classification.  Practical algorithms for these problems promise to provide applications outside of the physical sciences where quantum computation may be of great importance.  Our work provides an important step towards the goal of understanding the value of quantum computing to machine learning by rigorously showing that quantum speedups can be achieved for nearest--neighbor classification.

Consider the task faced by the U.S. postal service of routing over 150 billion pieces of mail annually~\cite{USPS}.
The sheer magnitude of this problem necessitates the use of software to automatically recognize the handwritten digits and letters that form the address of a recipient.
The nearest-neighbor algorithm is commonly used to solve tasks such as handwriting recognition due to its simplicity and high performance accuracy \cite{CH67}.
Nearest--neighbor classification algorithms also have found extensive use in detecting distant quasars using massive data sets from galactic surveys~\cite{BBM+06,BB10,GPK+10}.  These searches can be very computationally intensive as they often require analyzing over $60$ terabytes worth of data in some cases.  Hence finding distant quasars (which are luminous enough to be detected billions of light years away) requires very accurate classifiers that will not frequently make a false positive assignment.  Nearest--neighbor methods are often used because high assignment accuracy is needed given such massive data sets.  However, a major drawback of using classical nearest-neighbor classification is that it can be computationally expensive.

We show quantum computation can be used to polynomially reduce the query complexity of nearest--neighbor classification via fast quantum methods for computing distances between vectors and the use of amplitude estimation in concert with Grover's search.  This is no easy task as both of these methods require measurement but cannot be applied to a subroutine that exploits measurement in a non--trivial way. We address this problem by providing methods for removing measurements from the distance calculations and amplitude estimation so that the two can be used together coherently.  By combining these techniques, we achieve near--quadratic advantages over Monte--Carlo algorithms.  These improvements can help mitigate the main drawback of nearest--neighbor classification: its computational expense relative to other less exact methods.

Our quantum algorithms also reveal new possibilities for quantum machine learning that would not be practical for classical machine learning algorithms.  Our algorithm can be used to classify data that is output from a quantum simulator that cannot be efficiently simulated classically, similar to~\cite{WGFC14}.  This would allow fast algorithms to be constructed that rapidly classify a chemical using a large database of chemicals whose features can be computed on the fly using quantum simulation algorithms.  Unlike classical examples, the features (i.e. results from simulated experiments on the molecule) used for the classification do not need to be fixed but can vary from instance to instance.  This form of classification is not practical in a classical setting because it would require a prohibitively large database of features.  This is application is especially compelling for unsupervised machine learning wherein algorithms such as $k$--means clustering can be used to find appropriate classifications without human intervention.  We will see that quantum $k$--means clustering follows as a direct consequence of the methods that we develop for nearest--neighbor classification.

The paper is laid out as follows.  We review nearest-neighbor classification in~\sec{nn}.  We then discuss existing approaches to quantum classification and outline our quantum algorithms in and the oracles that they use in~\sec{oracle}.  We give our main results for the costs of performing nearest-neighbor classification using the ``inner product method'' and the ``Euclidean method'' in~\sec{ip} and~\sec{euclid} respectively.  We then provide numerical results in~\sec{numerics} that show that if our algorithms were to be applied to real world data sets methods then substantial errors in the estimates of the distance can be tolerated without degrading the assignment accuracy.  We compare our results to Monte--Carlo nearest--neighbor classification in~\sec{MC} and apply our results to $k$--means clustering in~\sec{kmeans} before concluding.  Proofs for all of our theorems are given in~\app{proofs}.

\section{Nearest-neighbor classification}\label{sec:nn}
The idea of nearest--neighbor classification is to classify based on experience.  The classifier attempts to classify a piece of data (known as a test vector) by comparing it to a set of training data that has already been classified by an expert (usually a human).  The test example is then assigned to the class of the training example that has the most similar features.  These features are often expressed as real--valued vectors, and can represent nearly any parameters of the data that the user wants to use as a basis for the classification.  The most similar vector is then calculated to be the training vector that is closest to the test vector using some appropriate distance metric (such as the Euclidean distance).  The test vector is then given the same classification that the expert gave to the closest training vector.

A major advantage of nearest--neighbor classification is that it tends to be very accurate given a large training set because as the data set gets larger.  This is because it becomes more likely that the test vector will be similar to a training vector in the database.  Both the accuracy and the simplicity of the method have caused nearest--neighbor classification to become a mainstay in classification and machine learning.  The three main drawbacks of nearest--neighbor approaches are that statistical outliers can lead to false positives, it is sensitive to errors on the training labels and  the algorithm can become prohibitively expensive for classical computers in cases where the database or the dimension of the training vectors  is large.  The former two problems can be dealt with by using $k$--nearest--neighbors classification, which is a straight forward generalization of nearest--neighbor classification, where the classification is generalized to a function of the $k$-nearest-neighbors.~\cite{FN75}.  Addressing the latter problem using quantum strategies is the focus of this paper.

As a concrete example of nearest--neighbor classification, consider the binary classification task of determining if a given unlabeled \emph{test} digit is \emph{even} or \emph{odd}.
The \emph{training} data consists of handwritten digits expressed as multidimensional feature vectors, each with a human-assigned label of either even or odd.
The entries of each feature vector $\vec{v}\in \mathbb{R}^N$, where $N$ is the number of features used to characterize the digit, are the pixel values that comprise the image.  \fig{example} shows an example of $25$ digits, each of which is represented by a $256$-dimensional feature vector of pixel values.

\begin{figure}[t]
\begin{minipage}{0.45\linewidth}
\includegraphics[width=0.5\columnwidth]{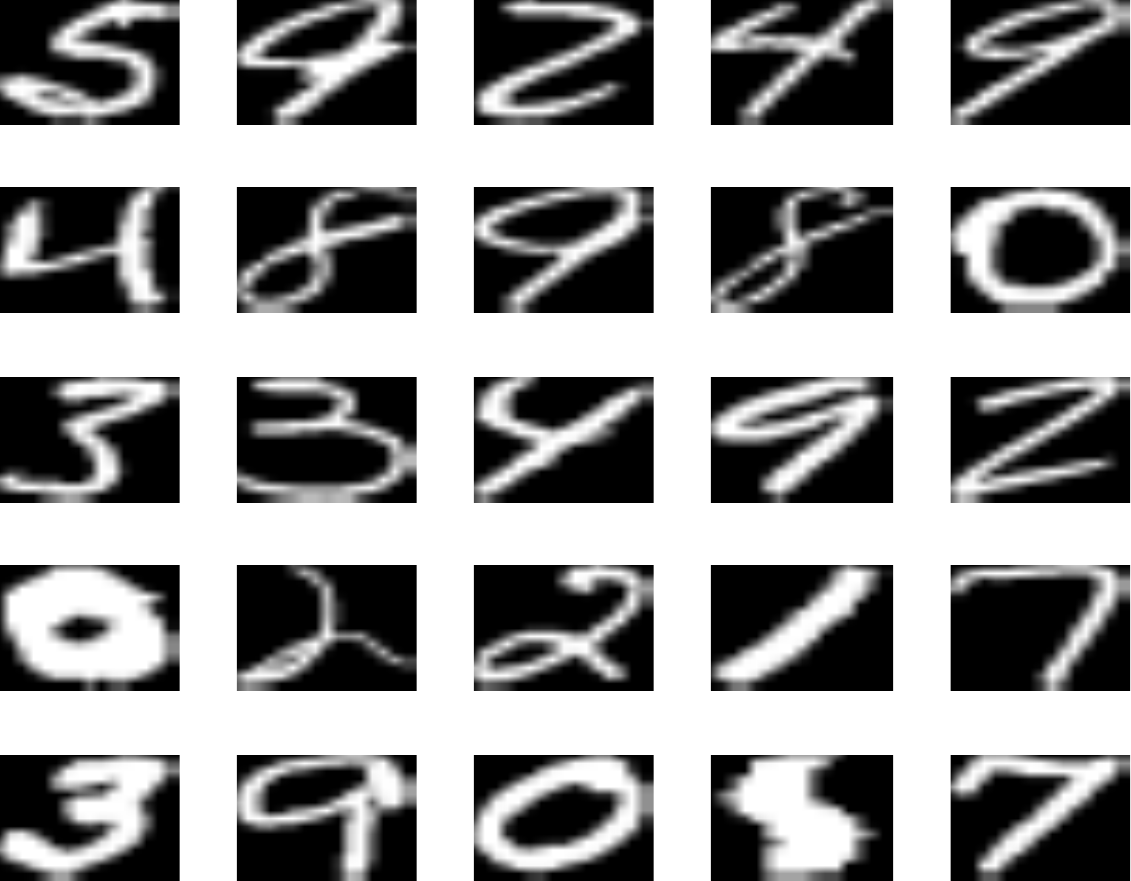}
\caption{An example of $25$ handwritten digits.  Each digit is stored as a $256$-pixel greyscale image and represented as a unit vector with $N=256$ features.\label{fig:example}}
\end{minipage}
\hspace{1mm}
\begin{minipage}{0.45\linewidth}
\includegraphics[width=\columnwidth]{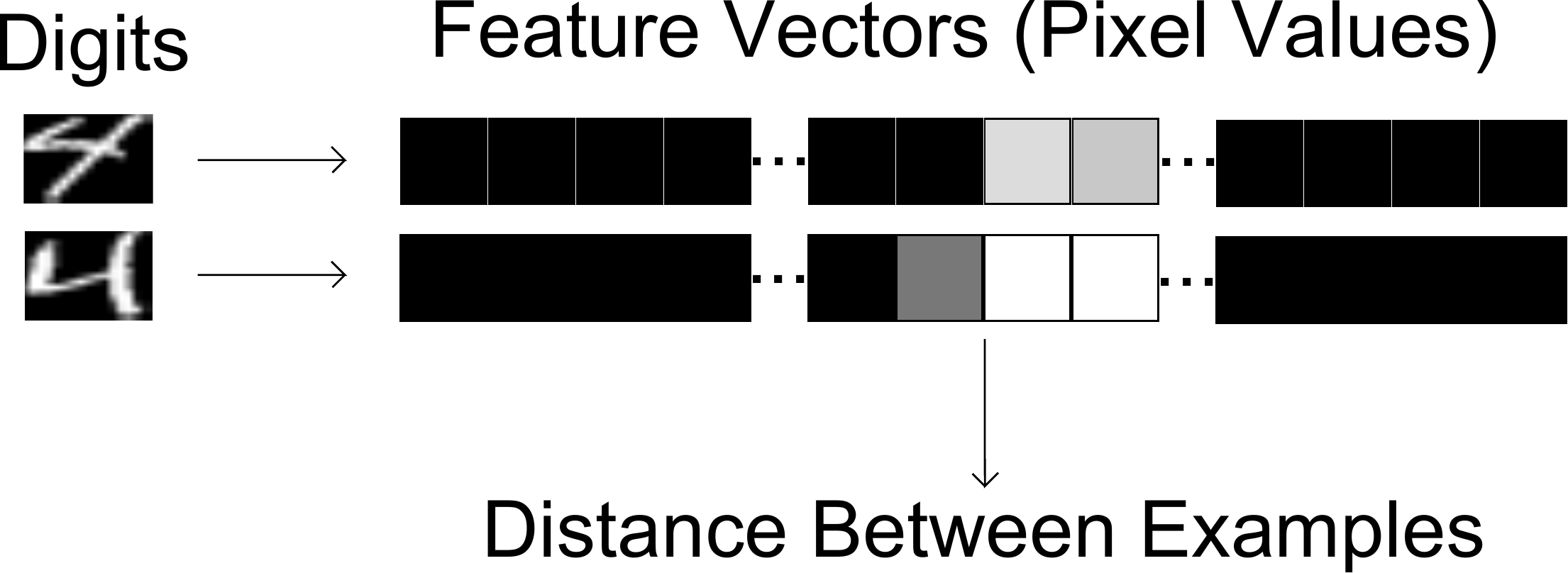}
\caption{Schematic process for processing raw data into feature vectors into distances for handwriting recognition.}
\end{minipage}
\end{figure}

We can divide the training set into two sets, or clusters, of vectors in $\mathbb{R}^N$, $\{A\}$ and $\{B\}$, such that $\{A\}$ contains only odd examples, $\{B\}$ contains only even examples,  and $|\{A\}|+|\{B\}|=M_A + M_B=M$.
The goal is to classify, or label, a given unlabeled \emph{test} point  $\vec{u}\in \mathbb{R}^N$ as $\vec{u}\in \{A\}$ or $\vec{u}\in \{B\}$.  More generally, vectors in $\mathbb{C}^N$ can also be used.
Here we take $N$ and $M$ to be large and cost the algorithm by the number of times the components of the vectors must be accessed during the classification procedure. 

The nearest-neighbor algorithm first computes the distance between the test vector $\vec{u}$ and each training vector $\vec{v}$ and then assigns $\vec{u}$ to the cluster that contains the closest vector to $\vec{u}$.
Specifically, it assigns $\vec{u}$ to $\{A\}$ if
$$
\min_{\vec{a}\in \{A\}}|\vec{u} -\vec{a}| \le \min_{\vec{b}\in \{B\}}|\vec{u} - \vec{b}|,
$$
for an appropriate distance metric $|\vec{u}-\vec{v}|$.  
Directly computing these quantities using the classical nearest-neighbor algorithm requires $O(NM)$ accesses to the components of these vectors..

\section{Quantum Nearest--Neighbor Classification}\label{sec:oracle}
Quantum computation shows promise as a powerful resource for accelerating certain classical machine learning algorithms~\cite{ABS06,DCH+08,LMR13,RML13} so it is natural to suspect that it may also be useful for quantum nearest--neighbor clasisfication.  
However, a major challenge facing the development of practical quantum machine learning algorithms  is the need for an oracle to return, for example, the distances between elements in the test and training sets.
Lloyd, Mohseni, and Rebentrost~\cite{LMR13} recently proposed a quantum algorithm that addresses this problem.
Namely, their algorithm computes a representative vector for each set, referred to as a \emph{centroid}, by averaging the vectors in $\{A\}$ and $\{B\}$, respectively.  The test point $\vec{u}$ is assigned to $\{A\}$ if the distance to the centroid of $\{A\}$, written as ${\rm mean}(\{A\})$, is smallest,
$$
|\vec{u} - {\rm mean}(\{A\})| \le |\vec{u} - {\rm mean}(\{B\})|,
$$
and $\{B\}$ otherwise.

Note that this algorithm, which we refer to as \emph{nearest--centroid classification} is a form of nearest-neighbor classification, where the nearest \emph{centroid} is used to determine the label, as opposed to the nearest training point.  If the number of clusters equals the number of points, then it reduces to nearest--neighbor classification. 
In practice, nearest--centroid classification can perform poorly because $\{A\}$ and $\{B\}$ are often embedded in a complicated manifold where the mean values of the sets are not within the manifold \cite{Lev05}.  
In contrast, nearest--neighbor classification tends work well in practice and often outperforms centroid--based classification but can be prohibitively expensive on classical computers~\cite{SCT02}.

Therefore, we present a quantum nearest--neighbor algorithm that assigns a point $\vec{u}$ to either cluster $\{A\}$ or $\{B\}$ such that both the probability of a faulty assignment and the number of quantum oracle queries is minimized.  
We consider two different ways of computing the distance within our algorithm: 
\begin{enumerate}
\itemsep 1pt
\parskip 0pt
\parsep 0pt
\item the inner product, $|\vec{u} - \vec{v}| = |\vec{u}||\vec{v}|-\vec{u}\cdot \vec{v}$, 
\item the Euclidean distance, $|\vec{u} - \vec{v}| = \sqrt{\sum_{i=1}^N (u_i - v_i)^2}$.
\end{enumerate}
It is easy to verify that these two measures are equivalent up to constants of proportionality for unit vectors, and so ideally both algorithms will provide exactly the same classification when applied to nearest--neighbor classification.
Our quantum algorithm overcomes the main drawback of the nearest-centroid approach in~\cite{LMR13}: low assignment accuracy in many real--world problems (as seen in~\sec{numerics} and~\app{numerics}).

Throughout, the test point is set to $\vec{v}_0:=\vec{u}$ and the training set consists of $\vec{v}_j$, for $j=1,\ldots,M$.
We assume the following: 
\begin{enumerate}
\item The input vectors are $d$--sparse, i.e., contain no more than $d$ non--zero entries.  Here $d$ can be as large as $N$, but if $d$ is chosen to be much larger than the true sparsity of the vectors then the performance of our algorithms will suffer.
\item Quantum oracles are provided in the form
\begin{align}
\mathcal{O}\ket{j}\ket{i} \ket{0}&:=\ket{j}\ket{i}\ket{v_{ji}},\nonumber\\
\mathcal{F}\ket{j}\ket{\ell}&:= \ket{j}\ket{f(j,\ell)},
\end{align}
where $v_{ji}$ is the $i^{\rm th}$ element of the $j^{\rm th}$ vector and $f(j,\ell)$ gives the location of the $\ell^{\rm th}$ non--zero entry in $\vec{v}_j$. 
\item The user knows an upper bound $r_{\max}$ on the absolute value of any component of $v_{ji}$.  
\item Each vector is normalized to $1$, for convenience (this is not necessary).  Non--unit vectors can be accommodated by rescaling the inner products computed using the inner product method, or by replacing $\vec{u}$ by the weighted average of $\vec{u}$ and $-\vec{u}$ adding fictitious vectors to $\vec{u}$ and $\vec{v}$ to make $\vec{u}$ and $\vec{v}$. 
\item The run time of the algorithm is dominated by the number of queries made to oracles $\mathcal{O}$ and $\mathcal{F}$.
\end{enumerate}

The main results that we find for performing nearest--neighbor classification under these assumptions are
\begin{enumerate}
\itemsep 1pt
\parskip 0pt
\parsep 0pt
\item The number of queries depends on $d r_{max}^2$ rather than on the feature dimension $N$ or the sparsity $d$ alone.
Thus, for practical applications the query complexity is typically independent of the number of features (i.e., $r_{\rm max}\propto 1/\sqrt{d}$).
\item The number of queries scales as $O(\sqrt{M}\log(M))$ rather than $M$ for nearest--neighbor classification.  The query complexity of our nearest--centroid algorithm (which uses the Euclidean method) can be independent of $M$.
\item Our algorithm can tolerate relatively large errors in distance calculations when applied to real--world classification problems.
\end{enumerate}

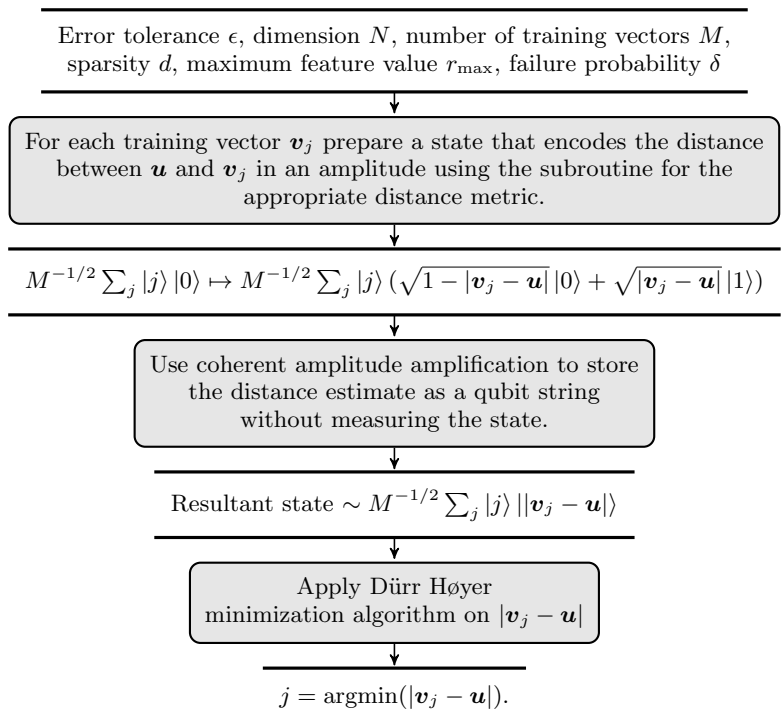
\begin{figure}[tb]
\centering
\makeatletter
\pgfdeclareshape{datastore}{
\inheritsavedanchors[from=rectangle]
\inheritanchorborder[from=rectangle]
\inheritanchor[from=rectangle]{center}
\inheritanchor[from=rectangle]{base}
\inheritanchor[from=rectangle]{north}
\inheritanchor[from=rectangle]{north east}
\inheritanchor[from=rectangle]{east}
\inheritanchor[from=rectangle]{south east}
\inheritanchor[from=rectangle]{south}
\inheritanchor[from=rectangle]{south west}
\inheritanchor[from=rectangle]{west}
\inheritanchor[from=rectangle]{north west}
\backgroundpath{
\southwest \pgf@xa=\pgf@x \pgf@ya=\pgf@y
\northeast \pgf@xb=\pgf@x \pgf@yb=\pgf@y
\pgfpathmoveto{\pgfpoint{\pgf@xa}{\pgf@ya}}
\pgfpathlineto{\pgfpoint{\pgf@xb}{\pgf@ya}}
\pgfpathmoveto{\pgfpoint{\pgf@xa}{\pgf@yb}}
\pgfpathlineto{\pgfpoint{\pgf@xb}{\pgf@yb}}
}
}
\makeatother
\usetikzlibrary{arrows}
\begin{tikzpicture}[
every matrix/.style={ampersand replacement=\&,column sep=1cm,row sep=0.3cm},
sink/.style={draw,thick,rounded corners,fill=gray!20,inner sep=.2cm},
datastore/.style={draw,very thick,shape=datastore,inner sep=.2cm},
to/.style={->,>=stealth',shorten >=1pt,semithick},
every node/.style={align=center}]
\matrix{
\& \node[datastore] (ang) {Error tolerance $\epsilon$, dimension $N$, number of training vectors $M$,\\
 sparsity $d$, maximum feature value $r_{\rm max}$, failure probability $\delta$}; \& \\
\& \node[sink] (appr) {For each training vector $\vec v_j$ prepare a state that encodes the distance\\  between $\vec u$ and $\vec v_j$ in an amplitude using the subroutine for the \\appropriate distance metric.}; \& \\
\& \node[datastore] (inia) {$M^{-1/2}\sum_j \ket{j}\ket{0} \mapsto M^{-1/2}\sum_j\ket{j}(\sqrt{1-|\vec v_j- \vec u|}\ket{0}+\sqrt{|\vec v_j- \vec u|}\ket{1})$}; \& \\
\& \node[sink] (rsea) {Use coherent amplitude amplification to store\\ the distance estimate as a qubit string \\ without measuring the state.}; \& \\
\& \node[datastore] (moda) {Resultant state $\sim M^{-1/2} \sum_j \ket{j}\ket{|\vec v_j- \vec u|}$}; \& \\
\& \node[sink] (design) {Apply D\"urr H\o yer\\ minimization algorithm on $|\vec v_j- \vec u|$}; \& \\
\& \node[datastore] (matrix) {$j={\rm argmin}( |\vec v_j- \vec u|)$.}; \& \\
};
\draw[to] (ang) --(appr);
\draw[to] (appr) --(inia);
\draw[to] (inia) --(rsea);
\draw[to] (rsea) --(moda);
\draw[to] (moda) --(design);
\draw[to] (design) --(matrix);
\end{tikzpicture}
\caption[Compilation algorithm]{High level description of the structure of the quantum nearest--neighbor algorithm.  Here $\ket{|\vec{v}_j - \vec{u}|}$ denotes a quantum state that holds a qubit string that represents the distance $|\vec{v}_j - \vec{u}|$.
}
\label{fig:alg}
\end{figure}

Implementation of oracular algorithms will require instantiation of the oracles, which are an abstraction of the many ways for algorithms to interact with data.
If the task is to classify chemicals, the oracle query could represent a call to an efficient quantum simulation algorithm that yields physical features of the chemicals~\cite{WBA11,WBC+13}.  
In other cases, the oracle query could represent accesses to a large quantum database that contains classical bit strings.  
One way to construct such a database is to use a quantum random access memory (qRAM)~\cite{GLM08}, however alternate implementations are possible.  
In this work, we assume oracles are provided and show how to minimize the number of queries to the oracle.

There are three phases to our quantum algorithm, which are laid out in~\fig{alg} and are given in pseudocode in~\app{pseudo}.  In the first phase we use the oracles $\mathcal{O}$ and $\mathcal{F}$ to prepare states of the form
$$
M^{-1/2}\sum_{j=1}^M\ket{j}\sum_{i=1}^N \left(\sqrt{1-v_{ji}^2/r_{\max}^2}\ket{0}+\sqrt{\frac{v_{ji}}{r_{\rm max}}}\ket{1}\right).
$$
This approach would be analogous to the Grover Rudolph state preparation procedure~\cite{GR02} if we were to measure the last qubit to be $1$.  We avoid such measurements because they would disallow subsequent applications of amplitude amplification.  Such states are then used in a circuit whose success probability, for each fixed $j$, encodes the distance between $\vec v_j$ and the test vector $\vec u$.

The second phase uses amplitude estimation~\cite{BHM+00} to estimate these success probabilities and store them in an ancilla register.  If we were to directly apply amplitude estimation to find these distances we would not only remove the possibility of using Grover's search to find the closest vector, so we use a form of amplitude estimation that we call coherent amplitude estimation to forgo the measurement at the price of preparing a logarithmically large number of copies of the state.  This results in a state that is, up to local isometries, approximately
$$
M^{-1/2}\sum_{j=1}^M\ket{j}\ket{|\vec{v}_j -\vec u|}.
$$

The final step of the algorithm is then to use Grover's search, in the form of the D\"urr H\o yer~\cite{DH96} minimization algortihm to search for the closest $\vec v_j$ to $\vec u$.  Since none of the prior steps uses measurement, their algorithm can be used without modification.  This means that we can achieve quadratic reductions in the scaling of the algorithm with respect to the error tolerance (due to amplitude estimation) and also with respect to $M$ (due to D\"urr H\o yer).

These steps are common to both of our approaches to nearest-neighbor classification.  We first discuss the inner product approach.  We then discuss an alternative approach that uses a circuit originally designed for implementing linear combinations of unitary operations to directly compute the Euclidean distance.  The Euclidean approach is conceptually more interesting, however,  because it naturally generalizes nearest--neighbor methods to nearest centroid methods.  In the case of nearest--neighbor classification, both methods provide exactly the same classification.  We provide both methods because the number of queries required to use the inner product method can sometimes be lower than that required by the Euclidean method.
\section{Inner Product Method}\label{sec:ip}
We first describe a quantum nearest--neighbor algorithm that directly computes the square of the inner product between two vectors to compute the distance.  
We show, somewhat surprisingly, that the required number of oracle queries does not explicitly depend on the number of features $N$. 
Rather, it depends implicitly on $N$ through $dr_{\max}^2$ and $\epsilon$.  
\begin{theorem}
Let $\vec{v}_0$ and $\{\vec{v}_j:j=1,\ldots,M\}$ be $d$--sparse unit vectors such that $\max_{j,i} |v_{ji}|\le r_{\max}$, then the task of finding $\max_j |\braketb{\vec{u}}{\vec{v}_j}|^2$ within error at most $\epsilon$ and  with success probability at least $1-\delta_0$ requires an expected number of combined queries to $\mathcal{O}$ and $\mathcal{F}$ that is bounded above by
$$
1080\sqrt{M} \left\lceil\frac{4\pi(\pi+1) d^2 r_{\max}^4}{\epsilon} \right\rceil\left\lceil\frac{\ln\left(\frac{81M(\ln(M) +\gamma)}{\delta_0} \right)}{2(8/\pi^2 -1/2)^2} \right\rceil,
$$\label{thm:innerproduct}
where $\gamma \approx 0.5772$ is Euler's constant.\label{thm:ip}
\end{theorem}
Two important scaling factors in the theorem should be emphasized.  
First, the scaling of the query complexity with $M$ is near--quadratically better than its classical analog.  
Second, if $r_{\max} \propto 1/\sqrt{d}$ then the scaling is independent of both $d$ and $N$.  
We expect this condition to occur when all input vectors have at least $\Theta(d)$ sparsity.

Note that the swap test gives the square of the inner product, rather than the inner product, as output.  This means that the sign of the inner product can be lost during the calculation.  In cases where the sign of the inner product is necessary for assignment, we can generalize the above method by transforming the quantum representation of the training vector $\ket{\vec{v}_\ell}$
\begin{align}
\ket{\vec{v}_\ell} \mapsto \frac{\ket{0}\ket{0^{\otimes \log_2 N}}+ \ket{1}\ket{\vec{v}_\ell}}{\sqrt{2}},\label{eq:vell}
\end{align}
and then using these states in~\thm{ip}.  
This allows direct estimation of the cosine distance and in turn the inner product.  We assume for simplicity throughout the following that $\ket{\vec{v}_\ell}$ represents the quantum state that directly corresponds to $\vec{v}_\ell$.  Translating these results to cases where the representation in~\eq{vell} is used is straight forward.

We prove~\thm{ip} by the following steps (see~\app{proofs} for more details).  
Assume that we want to compute the inner product between two states $\vec{v}_j$ and $\vec{v}_0:= \vec{u}$ and let $v_{ji}=r_{ji}e^{i\phi_{ji}}$, where $r_{ji}$ is a positive number.  This can be achieved using a coherent version of the swap test~\cite{BCW+01} on the states 
\begin{align}
&d^{-1/2} \sum_{i:v_{ji}\ne 0} \ket{i} \left(\sqrt{1-\frac{r_{ji}^2}{r_{\max}^2}}e^{-i \phi_{ji}}\ket{0}  + \frac{v_{ji}}{r_{\max}}\ket{1}\right)\ket{1}\nonumber,\\
&d^{-1/2}\sum_{i:v_{0i}\ne 0} \ket{i}\ket{1} \left(\sqrt{1-\frac{r_{0i}^2}{r_{ \max}^2}}e^{-i \phi_{0i}}\ket{0}  + \frac{v_{0i}}{r_{\max}}\ket{1}\right),
\label{eq:amppsi1}
\end{align}
which, as we show in~\app{proofs}, can be prepared using six oracle calls and two single--qubit rotations.
If the swap test is applied to these states and a probability of obtaining outcome `$0$', denoted $P(0)$, is found then
$$
|\braketb{\vec{u}}{\vec{v}_j}|^2=|\braketb{\vec{v}_0}{\vec{v}_j}|^2 = (2P(0)-1)d^2r_{max}^4.
$$
Statistical sampling requires $O(M/\epsilon^2)$ queries to achieve the desired error tolerance, which can be expensive if small values of $\epsilon$ are required.

We reduce the scaling with $\epsilon$ to $O(1/\epsilon)$ by removing the measurement in the swap test and applying amplitude estimation (AE)~\cite{BHM+00} to estimate $P(0)$ within error $\epsilon$, denoted $\tilde P(0)$.  
This can be done because the state preparation procedure and the measurement--free swap test are invertible.  
If a register of dimension $R$ is used in AE then the inference error obeys
$$
|P(0)-\tilde P(0)| \le \frac{\pi}{R} +\frac{\pi^2}{R^2}.
$$
Choosing $R$ to be large enough so that the error in AE is at most $\epsilon/2$ yields
\begin{equation}
R\ge \left\lceil\frac{4\pi(\pi+1)d^2r_{\max}^4}{\epsilon}\right\rceil.\label{eq:rbd}
\end{equation}

The inner product between non--unit vectors can easily be computed from the output of amplitude estimation because the inner product between two non--normalized vectors can be computed by rescaling the inner product of their unit--vector equivalents.  This can be done efficiently provided an oracle that yields the norm of each vector or by querying $\mathcal{O}$ and $\mathcal{F}$ $O(d)$ times.

The scaling with $M$ can also be quadratically reduced by using the maximum/minimum finding algorithm of D\"urr and H\o yer~\cite{DH96}, which combines Grover's algorithm with exponential search to find the largest or smallest element in a list.  
In order to apply the algorithm, we need to make the AE step reversible.  We call this form of AE \emph{coherent amplitude estimation}.

We achieve this by introducing a coherent majority voting scheme on a superposition over $k$--copies of the output of AE.  
AE outputs a state of the form $a\ket{y} +\sqrt{1-|a|^2} \ket{y^{\perp}}$, 
where $y$ is a bit--string that encodes $P(0)$ and $\ket{y^{\perp}}$ is orthogonal to $\ket{y}$.  
The median of $k$ bitstrings $x_k$ is computed coherently by $\mathcal{M}:\ket{x_1}\cdots \ket{x_k} \ket{0} \mapsto \ket{x_1} \cdots \ket{x_k} \ket{\bar{x}}$, where $\bar{x}$ is the median of $[x_{1},\ldots, x_k]$ (the mode could also be used). 
For this application, AE guarantees that $|a|^2\ge 8/\pi^2>1/2$ and Hoeffding's inequality shows that $\bar y= y$ with overwhelming probability if $k$ is sufficiently large.    
In particular, it is straightforward to show using the binomial theorem that we can write (see~\app{proofs})
\begin{align}
&\mathcal{M} (a\ket{y}+\sqrt{1-|a|^2}\ket{y^\perp})^{\otimes k}\ket{0}\nonumber\\
&\qquad=A\ket{\Psi}\ket{y} +\sqrt{1-|A|^2} \ket{\Phi;y'^\perp},\label{eq:derandomalg}
\end{align}
where $|A|^2>1-\Delta$ for $k\ge {\ln\left(\frac{1}{\Delta} \right)}/({2\left(8/\pi^2-\frac{1}{2}\right)^2)}$ and states $\ket{\Psi}$ and $\ket{\Phi; y'^\perp}$ are computationally irrelevant. We then use coherent majority voting to construct a $\sqrt{2\Delta}$--approximate oracle that maps $\ket{j}\ket{0} \mapsto \ket{j}\ket{\bar{y}}$. 
This approximate oracle is then used in the D\"urr H\o yer minimum finding algorithm.

We then make the pessimistic assumption that if the use of an approximate oracle leads to an erroneous outcome from the minimum finding algorithm even once then the whole algorithm fails.  Fortunately, since the number of repetitions of AE scales as $k\in O(\log(1/\Delta))$, this probability can be made vanishingly small at low cost.  Our final cost estimate then follows by multiplying, $k$, $R$, the costs of state preparation, and the number of iterations used in the D\"urr H\o yer algorithm.

\section{Euclidean Method}\label{sec:euclid}
We now describe a quantum nearest-neighbor algorithm that directly computes the Euclidean distance between $\vec{u}$ and the cluster \emph{centroids}, i.e., the mean values of the vectors within each cluster.   
This can be viewed as a step in a $k$--means clustering algorithm \cite{Llo82,LMR13}.  
We refer to this algorithm as a \emph{nearest-centroid} algorithm.
Our nearest--centroid algorithm differs substantially from that of~\cite{LMR13} in that (1) we normalize the computed distances, and (2) we consider a generalization to cases where each cluster is subdivided into $M'$ clusters that each contain $M_1,\ldots,M_{M'}$ vectors respectively.  
If $M'=M$ then the algorithm reduces to nearest--neighbor classification.

These differences help address two central problems of centroid--based classification.  
First, imagine that cluster $\{A\}$ is dense but cluster $\{B\}$ is sparse.  Then even if $|\vec{v}_0 -{\rm mean}(\{A\})| \le |\vec{v}_0 - {\rm mean}(\{B\})|$, it may be much more likely that the test vector $\vec{v}_0:=\vec{u}$ should be assigned to $\{B\}$ because the probability of a large deviation from the centroid is much greater for $\{B\}$ than $\{A\}$.  Alternatively, rescaling the distances can also be useful in cases where a faulty assignment to one class is less desirable than faulty assignment to the other class.
Normalizing the distance by the width of the cluster can help address these issues~\cite{THN+02}.  
We also show in~\app{ratio} that this assignment reduces to the likelihood ratio test under certain assumptions.
Second, if $\{A\}$ and $\{B\}$ are non--convex then the centroid of $\{A\}$ \emph{may actually be a point in $\{B\}$}.  
Segmenting the data into $M'$ smaller clusters can help address this issue.

The following theorem gives the query complexity for our quantum nearest--centroid algorithm.  The normalization of the distance by the cluster width can easily be omitted from the algorithm if desired.  Note that if $M'=M$ (which occurs when nearest--centroid classification reduces to nearest--neighbor classification) the intra--cluster variance is zero so the normalization step does not make sense in this case.  We set $\sigma_m=1$ in all these cases in order for the Euclidean method to be universally applicable to both nearest--neighbor and nearest--centroid classification.  Changing this normalization constant from $1$ to a class--dependent value may also be of use in cases where we want to bias faulty assignments, but we take $\sigma_m=1$ in the following for simplicity. 
\begin{theorem}
Let $\vec{v}_0$ and $\{\vec{v}_j^{(m)}:j=1,\ldots,M_m, m=1,\ldots M'\}$ be $d$--sparse unit vectors such that the components satisfy $\max_{m,j,i} |v_{ji}^{(m)}|\le r_{\max}$ and $\sigma_m=\frac{1}{M_m}\sum_{p=1}^{M_m} \|-\vec{v}_p^{(m)} +\frac{1}{M_m} \sum_j \vec{v}_j^{(m)}\|_2^2,$ if $M_m>1$ and $\sigma_m=1$ otherwise. The task of finding 
$$\min_m \left(\frac{\|\vec{v}_0 - \frac{1}{M_m} \sum_{j=1}^{M_m} \vec{v}_j^{(m)}\|^2_2}{\sigma_m}\right),$$ 
with error in the numerator and denominator bounded above by $\epsilon$ and  with success probability at least $1-\delta_0$, requires an expected number of combined queries to $\mathcal{O}$ and $\mathcal{F}$ that is bounded above by
$$
900\sqrt{M'} \left\lceil\frac{8\pi(\pi+1) d r_{\max}^2}{\epsilon} \right\rceil\left\lceil\frac{\log\left(\frac{81M'(\log(M') +\gamma)}{\delta_0} \right)}{2((8/\pi^2)^2 -1/2)^2} \right\rceil.
$$\label{thm:moment}
\end{theorem}

If $M' \in O({\rm polylog}(MN))$ and $dr_{\max}^2\in O(1)$ then the learning problem is efficient, which motivates the use of centroid--based classification for supervised learning problems where $\{A\}$ and $\{B\}$ can be partitioned into a union of a small number of disjoint training sets that are both unimodal and convex.  
Even if $M'\in \Theta(M)$ is required, then the query complexity of this method is at most comparable to the inner--product--based approach.


The proof of \thm{moment} follows similarly to that of \thm{ip}.  We use coherent amplitude estimation (AE) to find the numerator and the distance between $\vec{u}$ and the centroid as well as the intra--cluster variance.  
We then use a reversible circuit to divide the distance by the variance and use the D\"urr H\o yer algorithm~\cite{DH96} to find the minimum relative distance over all $M'$ clusters.  
The biggest conceptual difference is that in this case we do not use the swap test; instead we use a method from~\cite{RML13,CW12}.  

The result of~\cite{CW12} shows that for any unitary $V$ and transformation mapping $\ket{j}\ket{0} \mapsto \ket{j} \ket{\vec{v}_j}$, 
a measurement can be performed that has success probability
\begin{equation}
P(0)\propto \left|\sum_{j=0}^{M_m} |V_{j0}|^2 \vec{v}_j\right|^2,\label{eq:distance}
\end{equation}
for any $M_m$.
If we choose $\vec{v}_0=-\vec{u}$, then by~\eq{distance} and
$$|V_{j0}|=\begin{cases}\frac{1}{\sqrt{2}}, & j=0\\ \frac{1}{\sqrt{2M_m}}, & {\rm otherwise} \end{cases},$$
the probability of success gives the square of the Euclidean distance between $\vec{u}$ and the cluster centroid.  
Note that non--unit vectors can be accommodated by doubling the number of vectors and setting $\vec{v}_j \rightarrow \alpha_j\vec{v}_j -\sqrt{1-\alpha_j^2} \vec{v}_j$ where $\alpha_j\in [0,1]$.  
We show in~\app{proofs} that
$$\left| \vec{u} -\frac{1}{M_m} \sum_{j\ge 1} \vec{v}_j\right|^2=4d r_{\max}^2 P(0).$$
The operator $V$ can be implemented efficiently using techniques from quantum simulation, and the $\vec{v}_j$ are prepared using~\eq{amppsi1}.  The process of estimating the distance is therefore efficient (if the error tolerance is fixed).

The remainder of the procedure is identical to that of the inner--product method.  
The most notable technical difference is that the phase estimation procedure must succeed in both the distance and the intra--cluster variance calculations.  This results in the success probability in phase estimation dropping from at least $8/\pi^2$ to at least $(8/\pi^2)^2\approx 2/3$.
Thus, quantum nearest--centroid classification (based on the Euclidean method) requires more iterations than quantum nearest--neighbor classification (based on either method).

\begin{table}
\begin{tabular}{|c|c|c|}
\hline
Method & Typical cases & Atypical cases\\
\hline
Direct Calculation & $O(NM)$ & $O(NM)$\\
Monte--Carlo Calculation & $O(NMd^2r_{\max}^4)$ & $O\left(\frac{Md^2r_{\max}^4}{\epsilon^2}\right)$\\
Inner--Product Method & $O(\sqrt{NM}\log(M)d^2r_{\max}^4)$ & $O\left(\frac{\sqrt{M}\log(M)d^2r_{\max}^4}{\epsilon}\right)$\\
Euclidean Method & $O(\sqrt{NM'}\log(M')dr_{\max}^2)$ & $O\left(\frac{\sqrt{M'}\log(M')dr_{\max}^2}{\epsilon}\right)$\\
\hline
\end{tabular}
\caption{Query complexities for performing classification using the quantum and classical methods discussed here for examples where $M$ training vectors of dimension $N$ are used or $M'$ clusters are used.  Typical cases refer to examples where the training vectors are typical of random vectors, which implies that  $\epsilon\in\Theta(1/\sqrt{N})$ is needed to perform the classification accurately and atypical cases refer to examples where $\epsilon$ has no clear dependence on $N$.  In all cases $O(\log(NM/\epsilon))$ qubits are needed in addition to those required by arithmetic operations and instantiation of oracles.}
\end{table}

\begin{figure}[t!]
\begin{minipage}{0.45\linewidth}
\includegraphics[width=0.95\columnwidth]{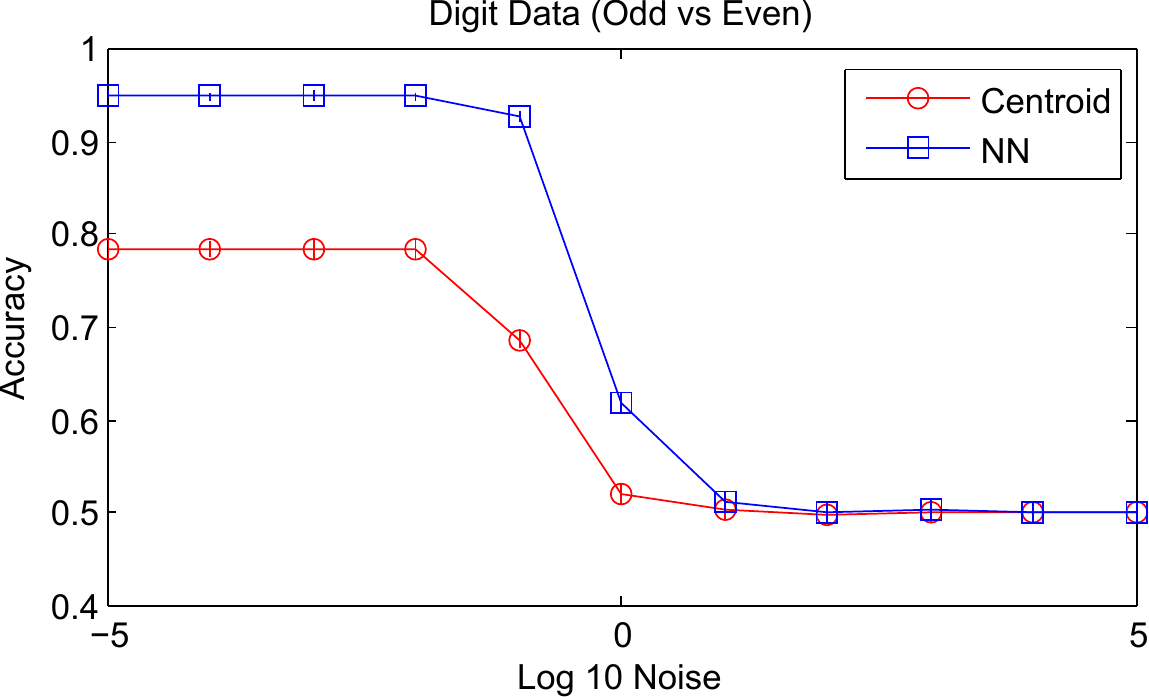}
\caption{Classification accuracy for digit data vs $\epsilon$ for cases where half the dataset is used for training.
\label{fig:classnoise}}
\end{minipage}
\hspace{1mm}
\begin{minipage}{0.45\linewidth}
\includegraphics[width=0.95\linewidth]{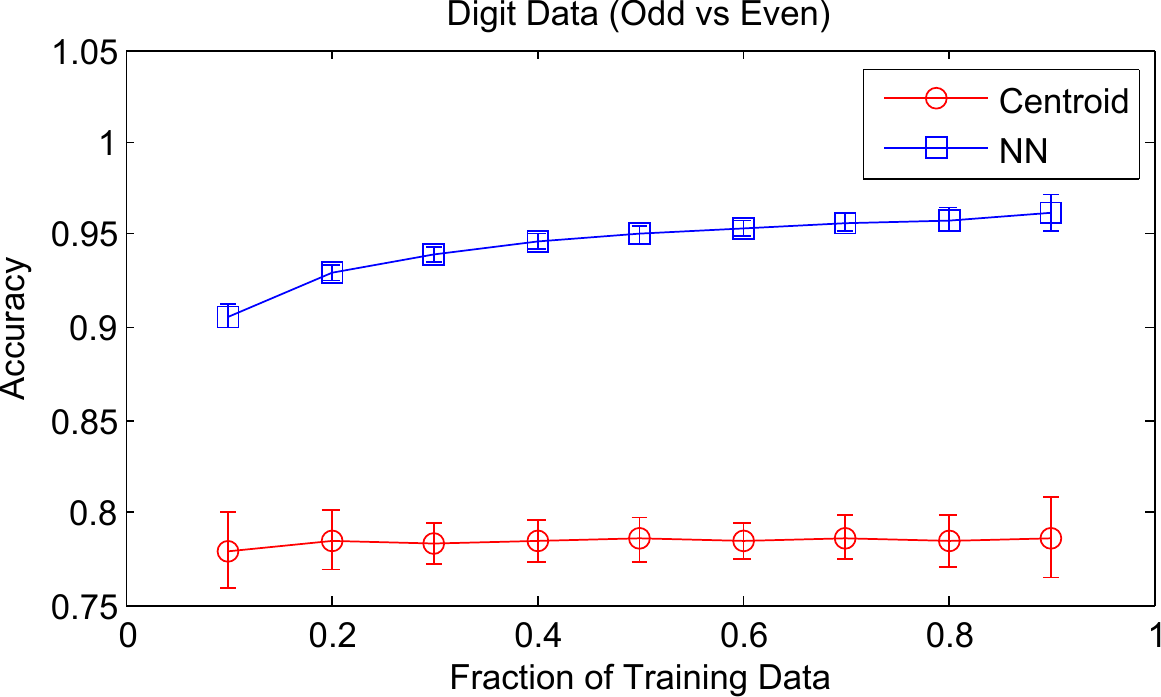}
\caption{Classification accuracy for digit data for fixed noise $\epsilon=10^{-5}$ as a function of the training data size.\label{fig:traintest}}
\end{minipage}
\end{figure}

\section{Numerical Experiments}\label{sec:numerics}
Although our algorithms clearly exhibit better scaling with $N$ and $M$ than direct classical computation of the nearest--neighbor of $\vec{u}$, the question remains whether $\epsilon$ needs to be prohibitively small in order ensure that the quality of the assignment is not impacted.  To this end, we evaluate the performance of our algorithms on several real--world tasks by simulating the noise incurred by using coherent amplitude estimation. 
As an example, we will consider classifying handwritten digits from the MNIST digits database \cite{MNIST} wherein the problem is:
given a training set of $M$ handwritten digits (see~\fig{example}) and their labels (even or odd), assign a label (even or odd) to an unlabeled test digit. 
Each digit is represented by a $256$-dimensional feature unit vector of pixel values \cite{MNIST}.
These pixel values are rescaled to have zero mean and unit variance over the data set before normalizing.
In all plots, error bars indicate standard deviation of the accuracy.

First, we compare the accuracy of our nearest--neighbor algorithm (NN) to the nearest--centroid algorithm (Centroid), as a function of noise $\epsilon$ in the distance computation.  Not that because the distances computed by the inner product method are logically equivalent to those computed using the Euclidean method, these results are valid for either algorithm.
For Centroid, we set $M'=1$.
\fig{classnoise} plots $\epsilon$ versus the accuracy of both NN (blue squares) and Centroid (red circles), for noise drawn from $\mathcal{N}(0,\epsilon)$ independently for each distance computation, where $\epsilon \in [10^{-5},10^5]$.
The accuracy is averaged across classification of $100$ random test examples using $M=2000$ training points.
In the low-noise regime, NN significantly outperforms Centroid by roughly $20\%$.
At $\epsilon \approx 0.1\approx 1/\sqrt{N}$, both algorithms exhibit significant loss in accuracy and eventually degrade to $50\%$ accuracy.  Both NN and Centroid tolerate relatively large errors without sacrificing classification accuracy.

The tolerance against noise up to size $O(1/\sqrt{N})$ is well--justified for high--dimensional systems by concentration of measure arguments~\cite{Led05}, so we anticipate that $\epsilon\in \Theta(1/\sqrt{N})$ should be appropriate for problems that lack underlying symmetry in the assignment set (such as even/odd classification).

Second, we study the effect of training data size on the performance accuracy of the two algorithms for a fixed noise rate $\epsilon=0.1$.
\fig{traintest} plots the training data size versus the performance accuracy of both NN (blue squares) and Centroid (red circles).
We vary the training set size by taking random fractions of $M=4000$ points, at fractions of $0.1,0.2,\ldots,0.9$.
For all training set sizes, NN significantly outperforms Centroid.
In addition, NN exhibits increasing performance accuracy as $M$ increases, from $84\%$ to $90\%$.
In contrast, Centroid's accuracy hovers around $73\%$, even as $M$ increases.

For the digit classification task, we estimate that accuracy $\alpha$ can be obtained using NN with a number of oracle queries that scales (for constant success probability) as $O(\sqrt{M})=O((1-\alpha)^{-5/4}\log((1-\alpha)^{-1}))$.
In contrast, the number of queries in the classical nearest-neighbor algorithm scales as $O((1-\alpha)^{-5/2})$.  
The centroid--based algorithm achieves at best $\alpha\approx 0.78$.
In addition, we find that $dr_{\max}^2\approx 2.8$ for this problem, indicating that the cost of state preparation will likely become negligible as $\alpha \rightarrow 1$ and in turn as $M\rightarrow \infty$.

While NN outperforms Centroid on the digit classification task, we find that for other tasks, outlined in~\app{numerics}, Centroid outperforms NN.  
However, in tasks where Centroid performs well, we find that both methods exhibit low classification accuracy.
This could indicate the need for more training data, in which case NN may begin to outperform Centroid as the amount of training data $M$ grows.  Such problems could also be addressed by adding more clusters to the nearest--centroid method (i.e. take $M'>1$).  In the digit classification case, it makes sense to take $M'=5$, so that the distance is calculated to each digit's centroid.  Although this may work well for digit classification, finding appropriate clusters for data in high--dimensional vector spaces can be costly as the problem is NP--hard in general (although heuristic algorithms often work well).  The key point behind that neither algorithm should be viewed as innately superior to the other.  They are simply tools that are appropriate to use in some circumstances but not in others.

We see from these results that nearest--neighbor classification is quite robust to error for the problem of handwriting recognition, whereas nearest centroid assignment does not perform as well.  We provide further numerical experiments for test cases related to diagnosing medical conditions in~\app{numerics}.  These results suggest that our quantum nearest--neighbor classification algorithm is not likely to be adversely affected by errors in the coherent amplitude estimation subroutine for problems with relatively small $N$.  We expect that such errors will be much more significant though in cases where $N$ is large (which occurs frequently in text classification problems) as argued in~\app{conc}.  This begs the question of whether direct calculation may actually provide advantages in some cases owing to the fact that there is no~$\epsilon$ dependence.  We address this issue along with the the question of how our algorithms compare to Monte--Carlo methods below.



\section{Comparison to Monte--Carlo Approaches}\label{sec:MC}
Although the natural analog of the quantum nearest--neighbor methods is direct calculation, it is important for us to consider how well the algorithm performs compared to Monte--Carlo algorithms for finding the distance.
The core idea behind these Monte--Carlo approaches is that it is seldom necessary to query the oracle to find all of the components of each $\vec v_j$ to compute the distance between $\vec v_j$ and $\vec u$.  This can substantially reduce the scaling of the cost of distance calculations with $N$ from $O(N)$ to $O(1)$ in some cases and tend to be particularly useful in cases where the training data is tightly clustered in high--dimensional spaces.  

A Monte--Carlo approximation to the inner product of two $d$--sparse vectors $a$ and $b$ can be found via the following approach.  First $N_c$ samples of individual components of $a$ and $b$ are taken.  If we assume that the locations where $a$ and $b$ are mutually non--zero are not known apriori then we can imagine that each vector is of dimension $D=\max(N,2d)$.  Let us denote the sequence of indexes to be $i_t$.  Then each component of the $D$--dimensional vector should be drawn with uniform probability (i.e., $p(i_t=x)=1/D$ for all $x$ in the union of the set of vectors that support $a$ and $b$).  Then an unbiased estimator of the inner product is given by
\begin{equation}
X = \frac{D}{N_c} \sum_{t=1}^{N_c}  a_{i_t}b_{i_t}.
\end{equation}
In particular, it is shown in~\cite{ESS+11} that
\begin{equation}
\mathbb{E}[X] = a^T b, \qquad \mathbb{V}[X] = \frac{1}{N_c}\left(D\sum_{i=1}^{N}a_{i_t}^2b_i^2-{(a^Tb)^2}\right)\in O\left(\frac{d^2 r_{\max}^4}{N_c}\right).
\end{equation}

Chebyshev's inequality therefore implies that for fixed vectors $a$ and $b$ that $N_c\in O(d^2r_{\max}^4\epsilon^{-2})$ is sufficient to guarantee that $X$ is a correct estimate to within distance $\epsilon$ with high probability.  Also, for random unit vectors, $D\sum_{i=1}^{N}a_{i_t}^2b_i^2-{(a^Tb)^2}=O(1)$ with high probability so typically the cost of the estimate will simply be $O(1/\epsilon^2)$.  The cost of nearest--neighbor classification is then $O(Md^2r_{\max}^4/\epsilon^2)$, which is (up to logarithmic factors) quadratically worse than our nearest--neighbor algorithm for cases where $dr_{\max}^2\in O(1)$.

A similar calculation implies that we can estimate the components of the mean vector to within error $\epsilon/N_c$ (which guarantees that the overall error is at most $\epsilon$).  To estimate this, we need to have the variance of each component of the vector.  Let $\{\vec{v}^{(m)}\}$ be a set of $d$--sparse unit vectors then
\begin{equation}
\mathbb{V}_m [\vec{v}_k^{(m)}] = \frac{1}{M}\sum_{m=1}^M (\vec{v}_{k}^{(m)} - \mathbb{E}_m[\vec{v}_{k}^{(m)}])^2 \le 4r_{\max}^2.
\end{equation}
Thus if we wish to estimate $\mathbb{E}_m [\vec{v}^{(m)}_j]$ within error $\epsilon/N_c$ (with high probability)
then it suffices to take a number of samples for each vector component (i.e., each $i_t$) that obeys
\begin{equation}
N_s\in O\left(\frac{r_{\max}^{2}N_c^2}{\epsilon^2} \right).
\end{equation}
Since there are $N_c$ different components, the total cost is $N_c N_s$ which implies that
\begin{equation}
{\rm Cost} \in O\left(\frac{r_{\max}^{2}N_c^3}{\epsilon^2} \right)\in O\left(\frac{d^6r_{\max}^{14}}{\epsilon^8} \right).\label{eq:montetime}
\end{equation}
Since $|a-b|_2 = a^Ta+b^Tb -2 a^Tb$, it follows that the Euclidean distance to the centroid can be computed using a number of queries that scales as~\eq{montetime}.  This simple argument suggests that the number of queries needed to a classical oracle to estimate the distance to the centroid is also efficient, for fixed $\epsilon$, using a classical sampling algorithm.  This cost in practice is prohibitively high and a more exact error analysis would be needed to find a better estimate of the true complexity of classical centroid-based classification.

An important question remains: is this scaling actually better than direct calculation?  We see in~\app{conc} that almost all unit vectors in $\mathbb{C}^N$ lie within a band of width $O(\sqrt{N})$ about any equator of the unit-hypersphere.  This means that if the vectors of both classes are evenly distributed then $\epsilon \in O(1/\sqrt{N})$ is needed in order ensure that nearest--neighbor methods can correctly assign the test vector.  In such cases, the cost of Monte--Carlo calculation of the nearest--neighbor of $\vec u$ is
$$
{\rm Cost} \in O\left({MNd^2 r_{\rm max}^4} \right),
$$
which is asymptotically equivalent to the cost of direct calculation if $dr_{\rm max}^2 \in O(1)$.  Therefore, if we are to regard random unit vectors as typical then Monte--Carlo methods do not typically offer asymptotic improvements over direct calculation.  Furthermore, we see that the near--quadratic improvement in the scaling with $\epsilon$ afforded by the use of oblivious amplitude amplification is needed in order to provide superior scaling with $N$ in such cases.  It is worth noting, however, that if the training data is atypical of Haar--random unit vectors then Monte--Carlo methods and in turn centroid--based classification may yield advantages over direct calculation.

\section{Application to $\bf{k}$--means clustering}\label{sec:kmeans}
These techniques can also be used to accelerate a single step of $k$--means~\cite{Mac67}, which is an unsupervised learning algorithm that clusters training examples into $k$--clusters based on their distances to each of the $k$ cluster centroids.  
Unlike supervised classification, such as nearest--neighbor classification, $k$--means clustering does not require human input in the form of pre--classified training data.  Instead, it seeks to learn useful criteria for classifying data that may not necessarily be apparent from the raw feature data.

Each of these clusters is uniquely specified by its centroid and the goal of the algorithm is to assign the training vectors to clusters such that the intra--cluster variance is minimized.  The algorithm for $k$--means clustering (also known as Lloyd's algorithm) is as follows: 
\begin{enumerate}
\item Choose initial values for the $k$ centroids.
\item For each training vector compute the distance between it and each of the $k$ cluster centroids.
\item Assign each training vector to the cluster whose centroid was closest to it.
\item Recompute the location of the centroids for each cluster.
\item Repeat steps $2-4$ until clusters converge or the maximum number of iterations is reached.
\end{enumerate}
The problem of optimally clustering data is known to be $\NP$--hard and as such finding the optimal clustering using this algorithm can be computationally expensive.  As also mentioned in~\cite{LMR13}, quantum computing can be leveraged to accelerate this task.  The query complexity of performing a single iteration of $k$--means (i.e. a single repetition of steps $2-4$) is given in the following corollary.

Our quantum algorithm for clustering deviates subtlely from the standard algorithm for clustering because the algorithm cannot easily output the cluster centroids and the cluster labels assigned to each training vector.  This is because the centroids need to be inferred using a process such as quantum state tomography.  In particular, if compressed sensing were used then the number of samples needed to learn the centroid within fixed accuracy scales as $O(kN^2\log N)$~\cite{FGL+12}.  This is prohibitively expensive given that $k$--means is an $O(kMN)$ algorithm classically.  

The following corollary gives the query complexity of performing an iteration of $k$--means using the complete set of labels for the members of each cluster to specify the centroids.  The vector representation of the centroids can be computed from these using the techniques of~\thm{moment}.  Proof is given in~\app{proofs}.
\begin{corollary}
Let $\vec{v}_0$ and $\{\vec{v}_j^{(m)}:j=1,\ldots,M_m, m=1,\ldots k\}$ be $d$--sparse unit vectors such that the components satisfy $\max_{m,j,i} |v_{ji}^{(m)}|\le r_{\max}$. The average number of queries made to $\mathcal{O}$ and $\mathcal{F}$ involved in computing the cluster assignments for each training vector in an iteration of $k$--means clustering, using distance calculations that have error $\epsilon$ and success probability $\delta_0$, is
$$
360M\sqrt{k} \left\lceil\frac{8\pi(\pi+1) d r_{\max}^2}{\epsilon} \right\rceil\left\lceil\frac{\log\left(\frac{81k(\log(k) +\gamma)}{\delta_0} \right)}{2((8/\pi^2) -1/2)^2} \right\rceil.
$$\label{cor:moment}
\end{corollary}
This shows that a step for $k$--means can be performed using a number of queries that scales as $O(M\sqrt{k}\log(k)/\epsilon)$, which is substantially better than the $O(kMN)$ scaling of the direct classical method if $kN\gg M$.

\section{Conclusions}\label{sec:conc}
We have presented quantum algorithms for performing nearest-neighbor classification and $k$--means clustering that promise significant reductions in their query complexity relative to over their classical counterparts.  
Our algorithms enable classification and clustering over datasets with both a high-dimensional feature space as well as a large number of training examples.
Computation of distances is extremely common in machine learning algorithms; we have developed two fast methods for computing the distance between vectors on a quantum computer that can be implemented coherently.
Finally, we have shown that our algorithms are robust to noise that arises from coherent amplitude estimation, perform well when applied to typical real--world tasks and asymptotically outperform Monte--Carlo methods for nearest--neighbor classification.

We find that quantum algorithms for machine learning can provide algorithmic improvements over classical machine learning techniques.
Our algorithms are a step toward blending fast quantum methods with proven machine learning techniques.  Further work will be needed to provide a complete cost assessment in terms of the number of elementary gate operations and logical qubits needed to practically achieve these speedups using a fault tolerant quantum computer.

Possibilities for future work include examining width/depth tradeoffs that arise when executing these algorithms on quantum computers and how well the algorithm performs in cases where the oracle represents a database.  Also it would be interesting to see whether these methods can be directly adapted to practical cases of characterizing data yielded by an efficient quantum circuit.  It remains an open question whether exponential speedups can be obtained by a quantum algorithm for supervised, unsupervised, or semi-supervised machine learning tasks.  Beyond computational speedups, a final interesting question is whether quantum computation allows new classes of learning algorithms that do not have natural classical analogs.  The search for inherently quantum machine learning algorithms may not only reveal potentially useful methods but also shed light on deeper questions of what it means to learn and whether quantum physics places inherent limitations on a system's ability to learn from data.

\acknowledgements{We thank Matt Hastings and Martin Roeteller for valuable comments and feedback.}


\appendix

\section{Additional numerical experiments}\label{app:numerics}
We evaluate the performance of our nearest--neighbor (NN) and nearest--centroid (Centroid) algorithms on several additional machine learning tasks.
A list of datasets, their respective training set sizes, and feature dimensions are listed in \tab{data}.  Each task is mapped to a binary classification problem (two classes).  
The data sets do not, in general, contain an equal number of training vectors per class.  We denote the number of training vectors in classes $A$ and $B$ to be $M_A$ and $M_B$, respectively.

The noise induced by inaccurate estimation of the distances is modeled by introducing Gaussian random noise with zero mean and variance $\epsilon^2$ and then clipping the result to the interval $[0,\infty)$.  Other distributions, such as uniformly distributed noise, gave qualitatively similar results. 
The features used in each data set can take on dramatically different value types.   
For example, the diabetes data set contains features such as patient age and blood pressure.  
In all tasks, we scale each feature to have zero mean and unit variance over the given data set.

We do not scale the vectors to unit length because the length of the vector is important for classification since points in one class are likely to be nearly co-linear with those in another class in low--dimensional spaces.  
Non--unit vectors can be easily accommodated by our algorithms by multiplying by the norms of the vectors in the inner--product based approach or by increasing the number of vectors used in the centroid approach.  This also means that $|\vec{v}_j|$ will typically be on the order of $\sqrt{N}$, which suggests that for the data sets that we consider $|\vec{v}_j|\in [1,10]$ is not unreasonable.  Hence we will refer to the regime where $\epsilon \le 1$ as the low--noise regime and $\epsilon \in (1,10]$ as the high--noise regime.

We first evaluate our algorithms on a standard machine learning benchmark commonly referred to as the ``half moon" dataset, which consists of two synthetically generated crescent-shaped clusters of points, as shown in ~\fig{halfmoon}.  
The dataset challenges classification algorithms since the convex hulls of the two ``moons" overlap and the mean value for each cluster (denoted by a star) sits in a region not covered by points.  This data set will be hard to classify with centroid--based methods (using one cluster) because $14.3\%$ of the data is closer to the centroid of the opposite set than to its own centroid.  This means that the accuracy of centroid--based assignment will be at most $85.7\%$.  In contrast, we expect nearest--neighbor classification to work well because the typical (Euclidean) distance between points is roughly $0.03$, whereas the two classes are separated by a distance of approximately $0.5$.  This means that NN should succeed with near $100\%$ probability, except  in cases where the training set size is very small.

\begin{figure}[t!]
\includegraphics[width=0.5\linewidth]{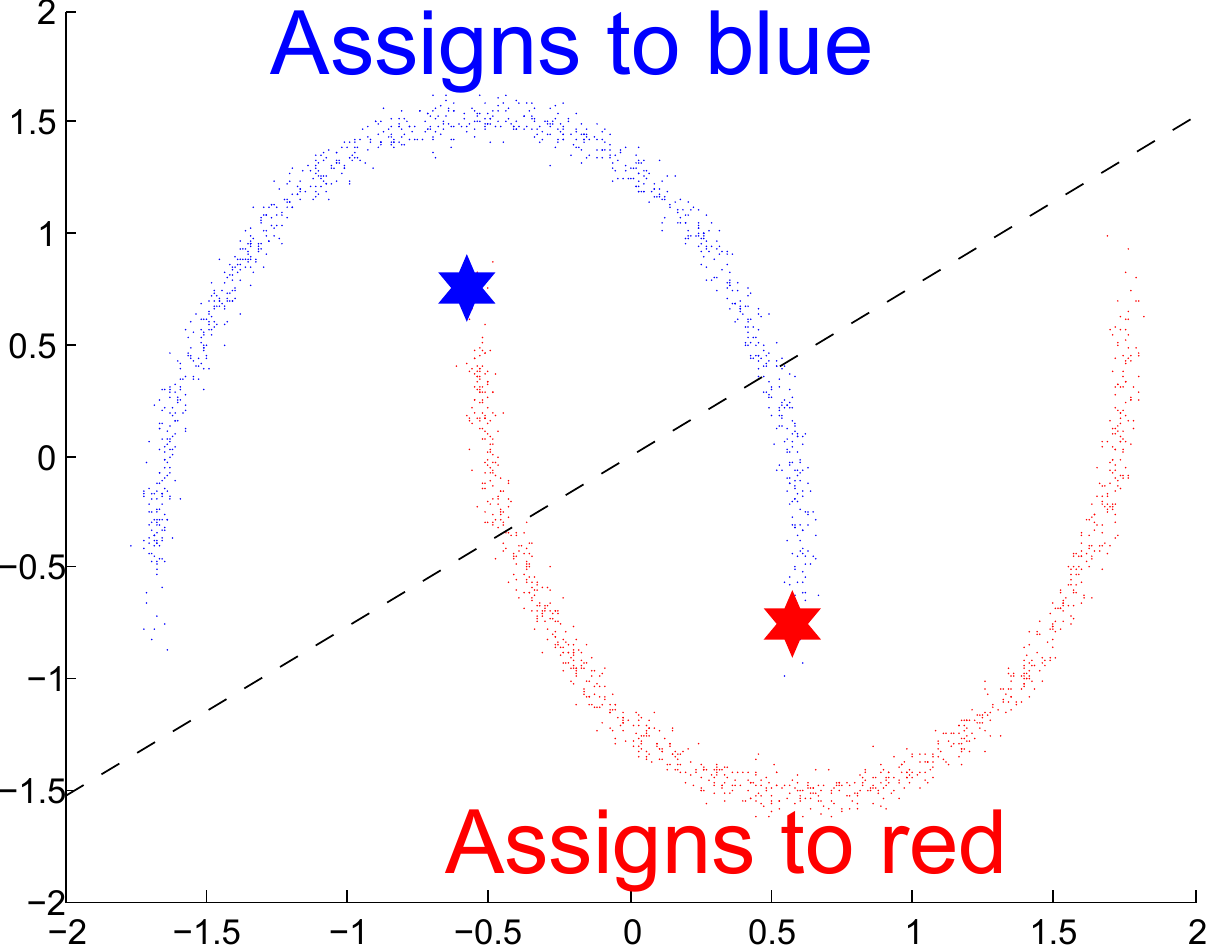}
\caption{Half--moon data set, vectors are unnormalized.  The two clusters of red and blue vectors correspond to the two classes used in the assignment set and the red and blue stars give the centroids of the corresponding cluster.  \label{fig:halfmoon}}
\end{figure}

\begin{figure}[t!]
\begin{minipage}{0.45\linewidth}
\includegraphics[width=\linewidth]{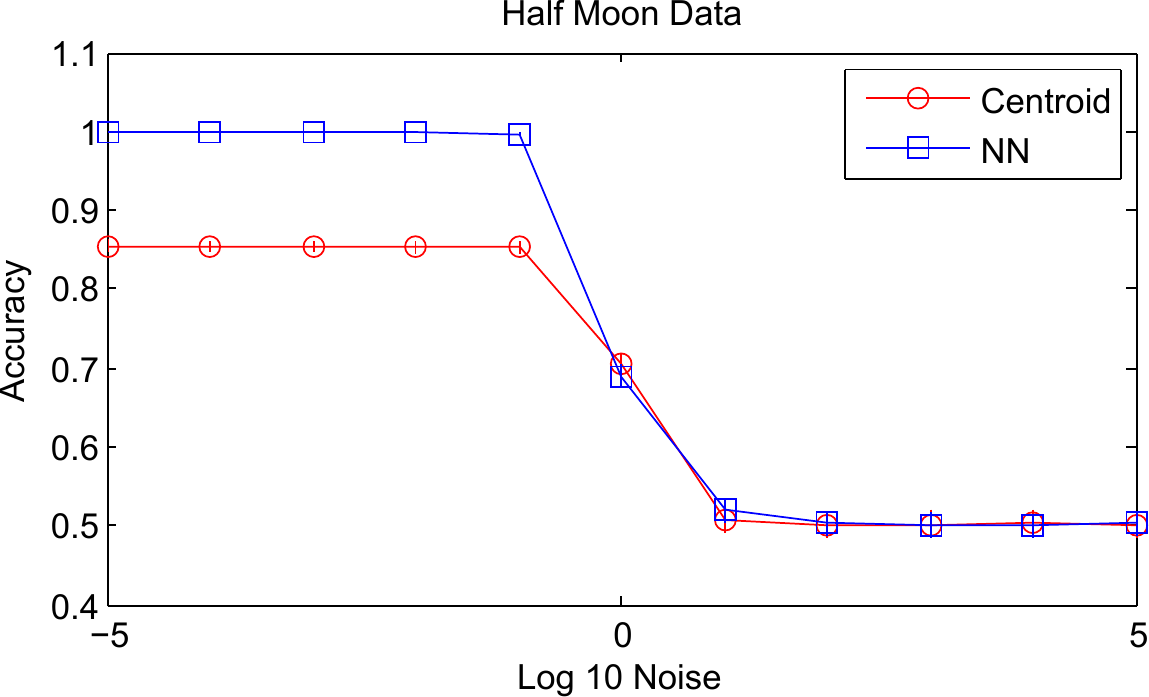}
\caption{Accuracy as a function of noise $\epsilon$ in distance computation for half--moon data\label{fig:halfmoonb}.  $50\%$ of the data was used to train the classifier and the remaining $50\%$ was used to test it.}
\end{minipage}
\hspace{0.5cm}
\begin{minipage}{0.45\linewidth}
\includegraphics[width=\linewidth]{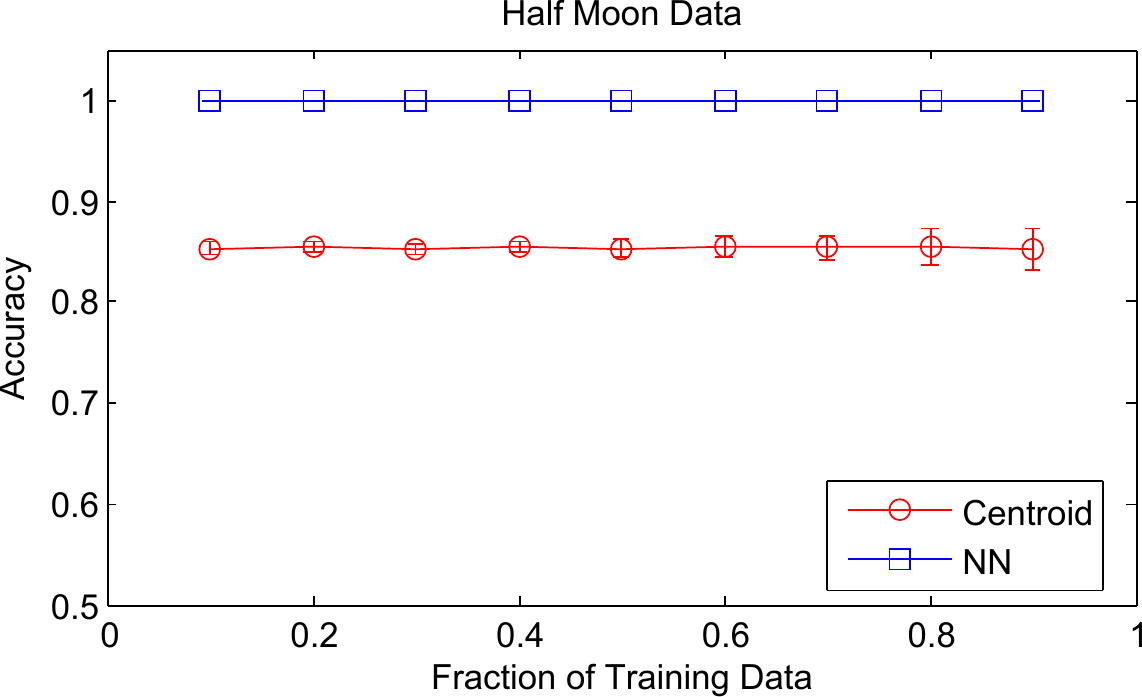}
\caption{Accuracy as a function of training data size for half--moon data\label{fig:halfmoonttb}.  Noise of $\epsilon=10^{-5}$ was used here.}
\end{minipage}
\end{figure}



In \fig{halfmoonb}, we plot the accuracies of our nearest--neighbor algorithm (NN; blue squares) and our nearest--centroid algorithm (Centroid; red circles) as functions of  noise $\epsilon$ in the distance computation.
NN significantly outperforms Centroid in the low--noise regime, exhibiting an accuracy near $100\%$ versus Centroid's $86\%$ accuracy.
As the noise level increases, the accuracy of both algorithms decays; however, in the low noise regimes, NN outperforms Centroid with statistical significance.  
At high noise levels, both algorithms decay to $50\%$ accuracy as expected.

\fig{halfmoonttb} shows accuracy as a function of training data size.
Here the training data size is taken to be a fraction, $f$,  of the $2000$ vectors in the set and the remaining fraction, $1-f$, of the $2000$ vectors was used to test the accuracy of the assignments.
Again, NN is almost always successful in classifying vectors; whereas Centroid achieves accuracies between $84$--$88\%$.  
Neither algorithm exhibits significant improvements in learning as the training set size is increased.
This behavior indicates the difficulty of this classification task for Centroid.

There are of course other methods that can be employed in order to boost the success probability of centroid--based classification.  The simplest is to cluster the data using a $k$--means clustering algorithm to subdivide each of the half moons into two or more clusters.  This semi--supervised approach often works well, but can be expensive for certain representations of the data~\cite{Jai10}.

The next tasks that we consider consist of determining whether a given disease was present or not based on patient data.  The diseases considered include breast cancer, heart disease, thyroid conditions, and diabetes. All data is taken from the UCL Machine Learning Repository~\cite{BL13}.
Details on the features and data size are given in \tab{data}.

\begin{figure}[t!]
\includegraphics[width=0.9\linewidth]{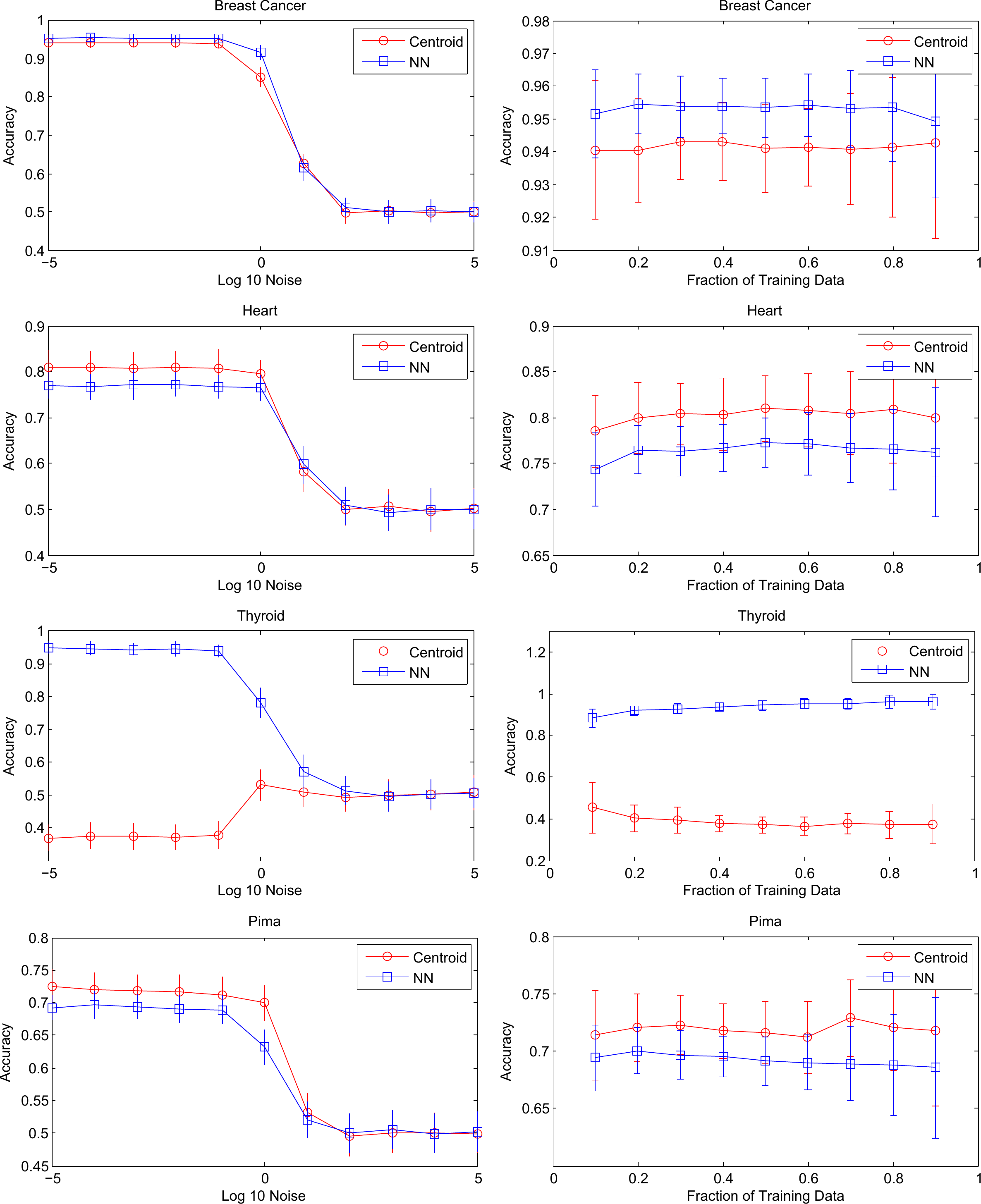}
\caption{(Left Column) Accuracy as a function of noise $\epsilon$ in the distance computation; (Right Column) Accuracy as a function of training set size for breast cancer (first row), heart disease (second row), thyroid (third row), and diabetes (fourth row) data. $50\%$ of the data is used for training and the remainder for testing for all data in the left column.  $\epsilon=10^{-5}$ is taken for all data in the right column.\label{fig:data}}
\end{figure}

\begin{table}[h!]
\begin{tabular}{|c|c|c|c|c|c|}
\hline
 & $N$, Number of Features & $M$, Number of Points & $M_A$ & $M_B$ & Year\\ \hline
Half Moon & $2$ & $2000$ & $1000$ & $1000$ & --\\
Breast Cancer \cite{MSW95,BL13} & $9$ & $683$ & $239$ & $444$ & $1992$ \\
Heart Disease (Statlog Data Set) \cite{BL13} &$13$ & $270$ &$120$ & $150$ & $1993$\\
Thyroid \cite{QCH87,BL13} &$5$ & $215$ & $150$ & $65$ & $1987$\\
Diabetes (Pima) \cite{BL13} &$7$ & $532$ & $177$ & $355$ &$1990$\\ \hline
\end{tabular}
\caption{Evaluation datasets.  Sizes of each data set for the conditions examined in~\fig{data}.}
\label{tab:data}
\end{table}

In some cases, we modified the data slightly. The breast cancer data, thyroid data, and the Pima diabetes study all contained instances of missing data.  
In each case we removed any vector that had a missing value.  We also removed boolean features from the thyroid and Pima diabetes data sets.

The left column of \fig{data} shows the accuracy of NN (blue squares) and Centroid (red circles) as a function of noise $\epsilon$ in the distance computations.
The first row shows the accuracies on the breast cancer data.
Both algorithms exhibit similarly high accuracies above $94\%$ in the low--noise regime, with NN outperforming Centroid with significance only at $\epsilon=1$.
In the extreme noise regime, NN performs just slightly better than random as expected.

In the second and last rows, the accuracies for heart disease and diabetes data are shown.
In these tasks, we find that in the low--noise regime, Centroid slightly outperforms NN, without statistical significance (except when $\epsilon=1$).
In the presence of high amounts of noise, both methods exhibit some learning; however, in all cases, learning is limited to around $55\%$.

In the third row, accuracy for the thyroid data is shown. 
NN exhibits significantly better accuracy of $90\%$ as compared to less than $40\%$ for Centroid.  In this case, the centroid--based algorithm performed worse than random guessing.  Poor accuracy is caused, in part, by our decision to divide the distance by the standard deviation in the distances as seen in \fig{thyroid}.  We found that the variance of the hypothyroid cases ($X_B$) was high enough that the mean of the training vectors that tested negative for thyroid conditions ($X_A$) was within one standard deviation of it.  In particular, $\sqrt{\mathbb{E}_{\vec{v}\in X_A}(|\vec{v} - {\rm mean}(X_B)|_2^2)/\sigma_B}\approx  0.49$ and $\sqrt{\mathbb{E}_{\vec{v}\in X_B}(|\vec{v} - {\rm mean}(X_A)|_2^2)/\sigma_A}\approx  4.4$.  Thus this test will incorrectly assign vectors from $X_A$ with high probability and correctly assign vectors from $X_B$ with high probability.  We therefore expect the accuracy to be roughly $30\%$ since the probability of drawing a vector from $X_B$ is roughly $65/215$.  This is close to to the observed accuracy of $37\% \pm 4\%$.

The data in~\fig{thyroid}, which forgoes normalizing the computed distances in Centroid, is devoid of these problems.  For low noise, Centroid succeeds roughly $86\%$ of the time and falls within statistical error of the NN data at $\epsilon\approx 1$.  Also, we observe that the assignment accuracy increases for both methods as more training data is used.  This is in stark contrast to the data in~\fig{data}; however,
this does not imply that the centroid--based method is actually performing well.  If we were to assign the data to class $A$ every time, regardless of the distance, we would succeed with probability $70\%$.  If Centroid is used, then the accuracy only increases by roughly $15\%$.  Also, since the two clusters strongly overlap, distance to the centroid is not a trustworthy statistic on which to base classification.  For these reasons, the use of Centroid to diagnose thyroid conditions, either with or without normalization, is inferior to using other methods.

\begin{figure}[t!]
\includegraphics[width=0.9\linewidth]{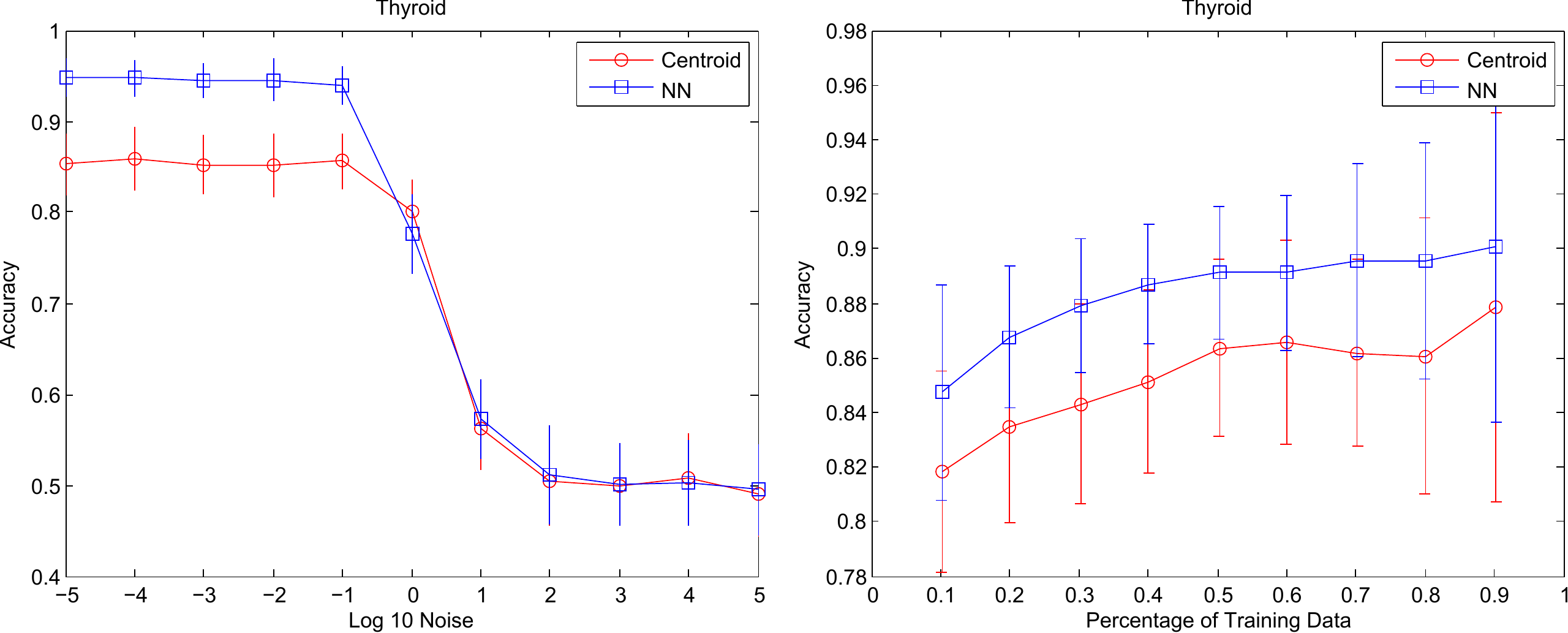}
\caption{Accuracy as a function of noise $\epsilon$ in the distance computation and the fraction of the total data that is used for training for thyroid data set where the normalization step in the distances has been omitted.  $50\%$ of the data is used for training and the remainder for testing in the left plot.  $\epsilon=10^{-5}$ is taken for all data in the right plot.\label{fig:thyroid}}
\end{figure}

The right column of \fig{data} shows the accuracy of our two algorithms as a function of training set size.
In the breast cancer task (first row), we see that both NN and Centroid exhibit little variation in accuracy as the amount of training data increases.
Similarly, in the heart disease and diabetes tasks (second and last rows), an increase in training data size does not imply significant increases in accuracy.
However, in the thyroid task, we see some differences in learning between NN and Centroid as the training data size increases.
NN's accuracy improves, from $85\%$ to $96\%$, while Centroid's accuracy decreases slightly. 

It is hard to determine in general why Centroid sometimes outperforms NN, but outliers in the data are frequently one reason.  
Outliers can cause problems for NN because it becomes increasingly likely as more training data is included that an outlier point from $X_B$ is close to any given element of $X_A$.  Thus increasing the training size can actually be harmful for certain nearest--neighbor classification problems.  Centroid is less sensitive to these problems because averaging over a data set reduces the significance of outliers.  Such problems can be addressed in the case of NN by using $k$--nearest--neighbor classifiers  instead of nearest--neighbor classification~\cite{FN75}.  Our quantum algorithms can be trivially modified to output the classes of each of the $k$ closest vectors (see~\app{kNN}).  Alternatively, such problems can also be addressed by using alternative machine learning strategies such as deep learning~\cite{Ben09}.


In summary, our numerical results indicate that classification accuracy, and in turn the best choice of algorithm, is highly dependent on the particular task and dataset.
While nearest--neighbor classification appears to be the preferred algorithm on most of the tasks presented here, in practice, a highly non-linear combination of classification algorithms is more commonly used \cite{Ben09}.
However, such classical approaches can be computationally expensive, in particular when classification over a large dataset is required.
Our quantum algorithms for classification offer the advantage of fast classification in conjunction with high performance accuracy, and may enable accurate classification of datasets that otherwise classically would not be possible.

\section{Proofs of main results}\label{app:proofs}
We present the proofs of \thm{ip} and \thm{moment} by way of a number of propositions that can be independently verified.
We begin with preliminary results that show that the state preparations used in our algorithms are efficient.  
We then review known results on the performance of the quantum minimum finding algorithm and amplitude estimation.  
We present our coherent majority voting scheme and variant of the swap test and provide  intermediate results needed to apply the D\"urr H\o yer algorithm and amplitude estimation coherently.  
We then use these results to prove \thm{ip}. 
Finally, we turn our attention to proving \thm{moment} which uses many of the same techniques used to prove \thm{ip}, but in addition requires the introduction of new methods for computing the distances to the cluster centroids and the  intra--cluster variance.

\subsection{Preliminary Results}
We begin by introducing a method to implement the operator $V$ which is needed for our nearest--centroid classification algorithm.
\begin{lemma}\label{lem:V}

A unitary $V$ such that 
 $$|V_{j0}|=\begin{cases}\frac{1}{\sqrt{2}}, & j=0\\ \frac{1}{\sqrt{2M}}, & {\rm otherwise.} \end{cases},$$
can be efficiently synthesized within error $O(\epsilon)$ on a quantum computer equipped with $H$ (Hadamard), $T$ ($\pi$/8) and CNOT gates.
\end{lemma}
\begin{proof}
Since $H$ is unitary and Hermitian it is a straightforward exercise in Taylor's theorem to show that for any $t\in\mathbb{R}$
\begin{equation}
e^{-iH^{\otimes m} t}=\openone \cos(t)-i H^{\otimes m} \sin(t).
\end{equation}
Thus if we choose $V$ to be $e^{-iHt}$ for some fixed value of $t$ then
\begin{equation}
V\ket{0} = \left(\cos(t)-i\frac{\sin(t)}{\sqrt{M+1}}\right)\ket{0} -i\frac{\sin(t)}{\sqrt{M+1}}\sum_{j>0} \ket{j}.
\end{equation}
The value of $t$ is found by setting $P(0)/P(j>0) = 1/M$, which yields
\begin{equation}
t=\sin^{-1}\left(\sqrt{\frac{M+1}{2M}} \right).
\end{equation}
Finally, $H$ can be made sparse efficiently by
\begin{equation}
(HT^2) H (T^6 H)= \left[\begin{array}{cc}0&e^{i\pi/4}\\ e^{-i\pi/4} &0  \end{array} \right],
\end{equation}
and hence $H^{\otimes n}$ can be transformed into a one--sparse matrix by applying this basis transformation to each qubit.  
One--sparse matrices can be efficiently simulated~\cite{ATS03,CCD+03,WBHS11} using gates $H$, $T$ and CNOT gates within error $O(\epsilon)$, completing the proof of the lemma.
\end{proof}

We also use the amplitude estimation result of Brassard et al.~\cite{BHM+00} to estimate the amplitude squared of a marked component of a quantum state, which we denote as $a$.  
The algorithm works by applying the phase estimation algorithm to an operator $Q$, which performs an iteration of Grover's algorithm where we wish to estimate the amplitude of the marked state.  
We provide a circuit for amplitude estimation in~\fig{AA}.
\begin{figure}[t!]
    \[
      \Qcircuit @R 1em @C 1.5em {
                    \lstick{\ket 0} &\qw {/}		&	\gate{F_L}	&\ctrl{1}			&\gate{F_L^\dagger}	&\qw &\rstick{\ket{y}}\\
		 		& \qw	{/}		&	\qw		&\gate{Q^j}			&\qw				&\qw
}				
    \]
\caption{Quantum circuit for amplitude estimation where $F_L$ is the $L$--dimensional Fourier transform and the controlled $Q^j$ operator applies $j$ Grover iterations to the target state if the top most register is $\ket j$\label{fig:AA}}
\end{figure}
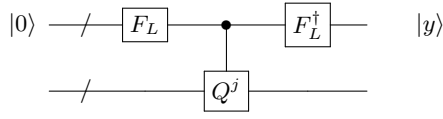
The following theorem shows that amplitude estimation can learn the resultant probabilities quadratically faster than statistical sampling.
\begin{theorem}[Brassard, H\o yer, Mosca and Tapp]
For any positive integers $k$ and $L$, the amplitude estimation algorithm of~\cite{BHM+00} outputs $\tilde{a}$ $(0 \le \tilde a \le 1)$ such that
$$
|\tilde{a}-a|\le 2\pi k \frac{\sqrt{a(1-a)}}{L}+\left(\frac{\pi k}{L}\right)^2
$$
with probability at least $8/\pi^2$ when $k=1$ and with probability greater than $1-1/(2(k-1))$ for $k\ge 2$.  
It uses exactly $L$ iterations of Grover's algorithm.  
If $a=0$ then $\tilde{a}=0$ with certainty, and if $a=1$ and $M$ is even, then $\tilde{a}=1$ with certainty.\label{thm:AE}
\end{theorem}

\begin{lemma}\label{lem:stateprep}
The state $\frac{1}{\sqrt{M}}\sum_{j=1}^{M} \ket{j}$ can be prepared efficiently and deterministically using a quantum computer.
\end{lemma}
\begin{proof}
Our proof proceeds by first showing a non--deterministic protocol for  preparing the state in question and then showing that amplitude amplification can be used make the state preparation process deterministic.

Let $m=\lceil \log_2( M+1)\rceil$, and let $\ket{M}$ be the computational basis state that stores $M$ as a binary string.
The proof follows from the fact that the following circuit
    \[
      \Qcircuit @R 1em @C 1.5em {
                        \lstick{\ket{0^{\otimes m}}} 	&	\gate{H^{\otimes m}}	&\ctrlo{3}	&\multigate{2}{{\rm CMP}}&\qw\\
		  \lstick{\ket{M}} 	&	\qw 	&\qw	&\ghost{{\rm CMP}}&\qw\\
		  \lstick{\ket{0}} 	&	\qw	&\qw	&\ghost{{\rm CMP}}&\meter\\
		  \lstick{\ket{0}}			&	\qw				&\targ		&\qw&\meter
}				
    \]
prepares the desired state given measurement outcome of $(0,0)$, which occurs with probability $\frac{M}{2^m}$.  
Here the operation ${\rm CMP}$ obeys
\begin{equation}
{\rm CMP}\ket{i}\ket{M}\ket{0}=\begin{cases} \ket{i}\ket{M}\ket{0},&  i\le M\\ \ket{i}\ket{M}\ket{1},& i> M \end{cases}.
\end{equation}
Here ${\rm CMP}$ can be implemented using the circuit in~\fig{CMP}.
\begin{figure}[t!]
    \[
      \Qcircuit @R 1em @C 1.5em {
                         	\lstick{\ket{i_p}}&	\ctrl{1}	& \qw		&	\ctrl{1}	&\qw\\
		 	\lstick{\ket{M_p}}&	\ctrlo{1}	&\qw		&	\ctrlo{1}	&\qw\\
			&	\targ		&\ctrl{1}	&	\targ		&\qw\\
			&	\ctrlo{-1}	&\targ		&	\ctrlo{-1}	&\qw
}				
    \]
\caption{Circuit for performing CMP illustrated for a single qubit inputs $\ket{i_p}$ and $\ket{M_p}$ after repeating this circuit $n$ times, the lower most register will contain $\ket{1}$ if $i>M$.\label{fig:CMP}}
\end{figure}
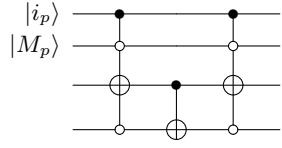
Hence the state can be prepared efficiently and with high probability if measurements are used.  
Also, note that the quantum control on the value of $M_p$ can be replaced with classical control in cases where a quantum superposition over different values of $M$ is not needed.

Since the success probability is known, the success probability can be boosted to certainty through amplitude amplification which requires $\Theta\left(\sqrt{\frac{2^m}{M}}\right)=\Theta(1)$ applications of CMP according to Theorem 4 of~\cite{BHM+00}.
This means that the measurement can be removed in the state preparation step without sacrificing the efficiency of the algorithm.
\end{proof}
Another important result is the method of D\"urr and H\o yer which is given as the following lemma~\cite{DH96}.
\begin{lemma}
[D\" urr H\o yer]  The expected number of Grover iterations needed to learn $\min \{y_i: i =1,\ldots,M\}$  is bounded above by
$$
\frac{45}{2} \sqrt{M}.
$$\label{lem:durr}
\end{lemma}
Nayak and Wu show that this algorithm is also near--optimal by providing a matching lower bound of $\Omega(\sqrt{M})$ for minimum finding (proven using the polynomial method) \cite{NW99}.

\subsection{Proof of Theorem 1}
In order to find the minimum value of a set of different quantum amplitudes using~\lem{durr}  we need to be able to perform iterations of Grover's algorithm using the result of~\thm{AE}.  
This cannot be done directly (with high probability) because the traditional approach to amplitude estimation is not reversible.  
We provide below a reversible algorithm that uses a coherent form of majority voting to obtain a reversible analog for algorithms like amplitude estimation.

\begin{lemma}
Let $\mathcal{A}$ be a unitary operation that maps $\ket{0^{\otimes n}}\mapsto \sqrt{a}\ket{y}+\sqrt{1-|a|} \ket{y^{\perp}}$ for $1/2 <|a_0|\le |a|\le 1$ using $Q$ queries then there exists a deterministic algorithm such that for any $\Delta>0$ there exists an integer $k$ and a state $\ket{\Psi}$ can be produced that obeys $\|\ket{\Psi}-\ket{0^{\otimes nk}}\ket{y}\|_2\le \sqrt{2\Delta}$ using a number of queries bounded above by
$$
2Q\left\lceil\frac{\ln(1/\Delta)}{2\left(|a_0|-\frac{1}{2} \right)^2}\right\rceil.
$$
\label{lem:expamp}
\end{lemma}
\begin{proof}
The basic idea behind the algorithm is to prepare $k$ copies of the state $\sqrt{a}\ket{y}+\sqrt{1-a} \ket{y^{\perp}}$, and then coherently compute the median via a reversible circuit and uncompute the $k$ resource states used to find the median of the values of $y$.  First, let $\mathcal{M}$ be a circuit that performs
\begin{equation}
\mathcal{M}:\ket{y_1}\cdots \ket{y_k} \ket{0} \mapsto \ket{y_1} \cdots \ket{y_k} \ket{\bar{y}},
\end{equation}
where we use $\bar{y_k}$ to denote the median.  This transformation can be performed by implementing a sort algorithm using $O(kn\log(k))$ operations and hence is efficient.   

The initial state for this part of the protocol is of the form
$$(\sqrt a\ket{y}+\sqrt{1-|a|}\ket{y^\perp})^{\otimes k}.$$
We can therefore partition the $k$--fold tensor product as a sum of two disjoint sets: the sum of states with median $y$ and another sum of states with median not equal to $y$.  We denote these two sums as $\ket{\Psi}$ and $\ket{\Phi}$ respectively, which is equivalent to expressing
\begin{equation}
(\sqrt a\ket{y}+\sqrt{1-|a|}\ket{y^\perp})^{\otimes k}:=A\ket{\Psi} +\sqrt{1-|A|^2} \ket{\Phi},\label{eq:psidef}
\end{equation}
for some value of $A$.
A direct consequence of~\eq{psidef} is that there exists (a possibly entangled state) $\ket{\Phi; y^\perp}$ such that
\begin{equation}
\mathcal{M} (\sqrt a\ket{y}+\sqrt{1-|a|}\ket{y^\perp})^{\otimes k}=A\ket{\Psi}\ket{y} +\sqrt{1-|A|^2} \ket{\Phi;y^\perp},\label{eq:derandomalg}
\end{equation}
where $A\ket{\Psi} +\sqrt{1-|A|^2} \ket{\Phi}:=(\sqrt a\ket{y}+\sqrt{1-|a|}\ket{y^\perp})^{\otimes k}$ and $\ket{\Psi}$ represents the subspace where the median is $y$ and $\ket{\Phi}$ is its compliment
Our goal is now to show that $|A|^2>1-\Delta$ for $k$ sufficiently large.

To see this, let us imagine measuring the first register in the computational basis.  
The probability of obtaining $p$ results, from the $k$ resulting bit strings, that are not $y$ is  given by the binomial theorem:
\begin{equation}
P(p) = \binom{k}{p}|a|^{p}|1-|a||^{k-p}.\label{eq:binomP}
\end{equation}
Now we can compute the probability that a measurement of the last register will not yield $y$ by observing the fact that in any sequence of measurements that contains more than $k/2$ $y$--outcomes, the median must be $k/2$.  Therefore the probability that the computed value of the median is not $y$ is at most the probability that the measured results contain no more than $k/2$ $y$ outcomes.  This is given by~\eq{binomP} to be
\begin{equation}
P(y^\perp) \le \sum_{p=0}^{\lfloor k/2\rfloor} \binom{k}{p} |a|^{p}|1-|a||^{k-p}.\label{eq:binomP2}
\end{equation}
Using Hoeffding's inequality on~\eq{binomP2} and $|a| \ge |a_0| >1/2$ we find that
\begin{equation}
P(y^{\perp})\le \exp\left(\frac{-2\left(k|a|-\frac{k}{2}\right)^2}{k} \right)= \exp\left({-2k\left(|a_0|-\frac{1}{2}\right)^2} \right).\label{eq:binomP3}
\end{equation}
Eq.~\eq{binomP3} therefore implies that $P(y^\perp)\le \Delta$ if
\begin{equation}
k\ge \frac{\ln\left(\frac{1}{\Delta} \right)}{2\left(|a_0|-\frac{1}{2}\right)^2}.\label{eq:kbound}
\end{equation}
Next, by applying $\mathcal{A}^{\dagger\otimes k}$ to the first register, we obtain
\begin{align}
\mathcal{A}^{\dagger \otimes k} \left(A\ket{\Psi}\ket{y} +\sqrt{1-|A|^2} \ket{\Phi;y^\perp}\right) &= \mathcal{A}^{\dagger \otimes k} \left(A\ket{\Psi}\ket{y} +\sqrt{1-|A|^2} \ket{\Phi}\ket{y}\right)+\mathcal{A}^{\dagger \otimes k}\left(\sqrt{1-|A|^2} \left(\ket{\Phi;y^\perp} -\ket{\Phi}\ket{y}\right)\right)\nonumber\\
&=\ket{0^{\otimes nk}}\ket{y}+\mathcal{A}^{\dagger \otimes k}\left(\sqrt{1-|A|^2} \left(\ket{\Phi;y^\perp} -\ket{\Phi}\ket{y}\right)\right).\label{eq:adaggererror}
\end{align}
Note that  $\braketb{y}{y^\perp}=0$ and hence $\ket{\Phi;y^\perp}$ is orthogonal to $\ket{\Phi}\ket{y^\perp}$.
If $|\cdot|$ is taken to be the $2$--norm then~\eq{adaggererror} gives that 
\begin{equation}
\left|\mathcal{A}^{\dagger \otimes k} \left(A\ket{\Psi}\ket{y} +\sqrt{1-|A|^2} \ket{\Phi;y^\perp}\right) -  \ket{0^{\otimes nk}}\ket{y}\right|\le \sqrt{2(1-|A|^2)}\le \sqrt{2\Delta},
\end{equation}
since $P(y^{\perp}):=1-|A|^2\le \Delta$ for $k$ chosen as per~\eq{kbound}.
The result then follows after noting that $k$ must be chosen to be an integer and that the total number of queries made to prepare the state is $2Qk$.
\end{proof}
\lem{expamp} shows that coherent majority voting can be used to remove the measurements used in algorithms such as amplitude estimation at the price of introducing a small amount of error in the resultant state.  
We can use such a protocol in the D\"urr H\o yer algorithm to find the minimum value of all possible outputs of the algorithm, as shown in the following corollary.
\begin{corollary}
Assume that for any $j=1,\ldots M$, a unitary transformation $\ket{j}\ket{0^{\otimes n}}\mapsto \ket{j}\left(\sqrt{a}\ket{y_j}+\sqrt{1-|a|} \ket{y^{\perp}_j}\right)$ for $1/2 <|a_0|\le |a|\le 1$ can be performed using $Q$ queries then the expected number of queries made to find $\min_j y_j$ with failure probabilty at most $\delta_0$ is bounded above by
$$
90\sqrt{M}Q\left\lceil\frac{\ln\left(\frac{81M(\ln(M) +\gamma)}{\delta_0} \right)}{2\left(|a_0|-\frac{1}{2} \right)^2}\right\rceil.
$$
\label{cor:expamp}
\end{corollary}
\begin{proof}

\lem{durr} states that at most $45\sqrt{M} /2$ applications of Grover's search are required, which requires $45\sqrt{M}$ queries to an (approximate) oracle that prepares each $y_j$ since two queries are required per Grover iteration (two are required to perform the reflection about the initial state and none are required to reflect about the space orthogonal to the marked state here).  \lem{expamp} therefore says that the cost of performing this portion of the algorithm is
\begin{equation}
N_{\rm queries}\le 90\sqrt{M}Q\left\lceil\frac{\ln(1/\Delta)}{2\left(|a_0|-\frac{1}{2} \right)^2}\right\rceil.\label{eq:nqueries}
\end{equation}

Next, we need to find a value of $\Delta$ that will make the failure probability for this approximate oracle at most $\delta_0$.    
Now let us assume the worst case scenario that if the measurement of $y_j$ fails to output the desired value even once, then the entire algorithm fails.  We then upper bound the probability of failure by summing the probability of failure in each of the steps in the search.  Assuming that the algorithm is searching for an element of rank at least $r$ (in the sense of D\"urr H{\o}yer~\cite{DH96}) then the number of calls to the oracle yielding $y_j$ is at most~\cite{BHM+00}
$$
9\sqrt{\frac{M}{r-1}}.
$$
This means that the amplitude of the erroneous component of the state (using subadditivity of quantum errors) is at most
$$
9\sqrt{\frac{\Delta M}{r-1}}.
$$
The worst case scenario is that the algorithm must search through all $M$ entries (this is extremely unlikely if $M$ is large because the average complexity is $O(\sqrt{M})$).  This means that the probability of at least one failed observation occuring  is at most
\begin{equation}
\sum_{r=2}^M \frac{81\Delta M}{r-1}=81MH_{M-1}\le 81M(\ln(M)+\gamma).
\end{equation}
Here $H_{M-1}$ is the $M-1^{\rm th}$ harmonic number and $\gamma$ is Euler's constant.
Therefore if we want the total probability of error to be at most $\delta_0$ then it suffices to choose
\begin{equation}
\Delta=\frac{\delta_0}{81M(\ln(M)+\gamma)}.\label{eq:Deltabd}
\end{equation}
Then combining~\eq{nqueries} and~\eq{Deltabd} gives us that the average query complexity obeys
\begin{equation}
N_{\rm queries}\le 90\sqrt{M} Q\left\lceil\frac{\ln\left(\frac{81M(\ln(M) +\gamma)}{\delta_0} \right)}{2(|a_0| -1/2)^2} \right\rceil.\label{eq:nqueries2}
\end{equation}
\end{proof}
Note that we want to maximize the value of $\sin^2(\pi y_j/R)$ that is yielded by the amplitude estimation algorithm for each $j$.  This maximization is equivalent to minimizing $|R/2-y_j|$.  Given that $y_j$ is returned coherently by our de--randomized amplitude estimation circuit, a measurement--free circuit can be used that computes $|R/2-y_j|$ for any input $y_j$.  This requires no further oracle calls.  Hence~\cor{expamp} applies to our circumstances with no further modification.

The results of \thm{ip} and \thm{moment} follow directly from~\cor{expamp} by substitution of appropriate values of $Q$ and $|a_0|$.  The remaining work focuses on devising an appropriate state preparation algorithm that can be used in~\thm{AE}.

\begin{lemma}\label{lem:density}
Let $\vec{v}_j$ be $d$--sparse and assume that the quantum computer has access to $\mathcal{O}$ and $\mathcal{F}$ then a unitary  transformation exists that can be implemented efficiently using $3$ oracle calls and, for all $j$, maps
$$
\ket{j}\ket{0}\mapsto\frac{1}{\sqrt{d}}\ket{j} \sum_{i=1}^{d} \ket{f(j,i)}\left(\sqrt{1-\frac{r_{jf(j,i)}^2}{r_{j\max}^2}}e^{-i\phi_{jf(j,i)}}\ket{0}+\frac{r_{jf(j,i)}e^{i\phi_{jf(j,i)}}}{r_{j\max}}e^{i\phi_{jf(j,i)}}\ket{1} \right).
$$
\end{lemma}
\begin{proof}
We begin by preparing the state
\begin{equation}
\frac{1}{\sqrt{{d}}} \sum_{i=1}^{d} \ket{j}\ket{i}\ket{0}\ket{0},
\end{equation}
which can be prepared reversibly and efficiently by applying \lem{stateprep}.  

The next step is to apply the oracle $\mathcal{F}$ to the result, this performs
\begin{equation}
\frac{1}{\sqrt{{d}}} \sum_{i=1}^{d} \ket{j}\ket{i}\ket{0}\ket{0}\mapsto \frac{1}{\sqrt{{d}}} \sum_{i=1}^{d} \ket{j}\ket{f(j,i)}\ket{0}\ket{0}.
\end{equation}
Then querying $O$ implements
\begin{equation}
\frac{1}{\sqrt{{d}}} \sum_{i=1}^{d} \ket{j}\ket{f(j,i)}\ket{0}\ket{0} \mapsto \frac{1}{\sqrt{{d}}} \sum_{i=1}^{d}\label{eq:b20} \ket{j}\ket{f(j,i)}\ket{v_{jf(j,i)}}\ket{0}.
\end{equation}

By applying $R_y(2\sin^{-1}(r_{jf(ji)}/r_{j\max}))$ on the final qubit in~\eq{b20}, we obtain

\begin{equation}
\frac{1}{\sqrt{{d}}} \sum_{i=1}^{d} \ket{j}\ket{f(j,i)}\ket{v_{jf(j,i)}}\ket{0}\mapsto\frac{1}{\sqrt{d}}\ket{j} \sum_{i=1}^{d} \ket{f(j,i)}\ket{v_{jf(j,i)}}\left(\sqrt{1-\frac{r_{jf(j,i)}^2}{r_{j\max}^2}}\ket{0}+\frac{r_{jf(j,i)}}{r_{j\max}}\ket{1} \right).\label{eq:b21}\end{equation}
The result then follows by applying $R_z(2\phi_{jf(j,i)})$ to the last qubit in~\eq{b21} and using  $O^{\dagger}$ to clean the ancilla register containing $\ket{v_{jf(j,i)}}$.  Three queries are used in this process.
\end{proof}


Next we use the swap test to provide a method to compute the inner product between two vectors.    The test is implemented by the circuit in~\fig{swap} for arbitrary states $\ket{\phi}$ and $\ket{\psi}$.
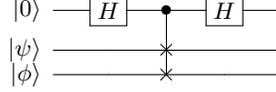
\begin{figure}
    \[
      \Qcircuit @R 1em @C 1.5em {
                        \lstick{\ket{0}} 	&	\gate{H}	&\ctrl{2}	&\gate{H}&\qw 	\\
		  \lstick{\ket{\psi}}	&\qw&\qswap		&	\qw&\qw\\
		  \lstick{\ket{\phi}}	&\qw&\qswap		&	\qw&\qw
}				
    \]
\caption{The swap test~\cite{BCW+01}.  The probability of measuring the top qubit to be zero is $1/2 + |\!\braketb{\phi}{\psi}\!|^2/2$, which allows statistical testing to be used to efficiently discriminate the states.\label{fig:swap}}
\end{figure}
The resultant state before measurement is
$$
\frac{1}{2} \ket{0}(\ket{\phi}\ket{\psi} + \ket{\psi}\ket{\phi})+\frac{1}{2} \ket{1}(\ket{\phi}\ket{\psi} - \ket{\psi}\ket{\phi}),
$$
and the probability of measuring the first qubit to be $1$ is $1/2 - |\braketb{\phi}{\psi}|^2/2$.  We do not ignore this measurement since we want to use the swap test within the Grover iterations used in~\thm{AE}.

\begin{lemma}
For any fixed $\epsilon>0$ and any pair of $d$--sparse unit vectors $\vec{u}\in \mathbb{C}^n$ and $\vec{v}_j\in \mathbb{C}^n$ a state of the form $\sqrt{|A|}\ket{\Psi}\ket{y}+ \sqrt{1-|A|}\ket{\Phi;y_\perp}$ can be efficiently prepared where $y$ encodes $|\braketb{u}{v_j}|^2$ within error $\epsilon$ and $|A| \ge 8/\pi^2$ using a number of queries that is bounded above by
$$Q\le 12\left\lceil\frac{4\pi(\pi+1)d^2r_{0\max}^2 r_{j \max}^2}{\epsilon}\right\rceil,$$
where $|\braketb{i}{v_j}|\le r_{j\max}$ for any $i\ge 0$.\label{lem:ipamp}
\end{lemma}

\begin{proof}
The state $\ket{\phi}\otimes \ket{\psi}$ can be prepared using \lem{density} and $n$ CNOT gates where
\begin{align}
\ket{\psi}&:=\frac{1}{\sqrt{d}}\ket{j} \sum_{i=1}^d \ket{f(j,i)} \ket{0^{\otimes n'}} \left(\sqrt{1-\frac{r_{jf(j,i)}^2}{r_{j \max}^2}}e^{-i \phi_{jf(j,i)}}\ket{0}  + \frac{v_{jf(j,i)}}{r_{j \max}}\ket{1}\right)\ket{1}.\label{eq:amppsi}\\
\ket{\phi}&:=\frac{1}{\sqrt{d}}\ket{j} \sum_{i=1}^d \ket{f(0,i)} \ket{0^{\otimes n'}}\ket{1} \left(\sqrt{1-\frac{r_{0f(0,i)}^2}{r_{0 \max}^2}}e^{-i \phi_{0f(0,i)}}\ket{0}  + \frac{v_{0f(0,i)}}{r_{0 \max}}\ket{1}\right).\label{eq:ampphi}
\end{align}
It is important to note that the qubit that encodes the component $v_{jf(j,i)}$ and $v_{0f(0,i)}$ differs between $\ket{\psi}$ and $\ket{\phi}$.  This is done in order to ensure that the inner product between the two vectors corresponds to the true inner product between the vectors.  Specifically,
\begin{equation}
\braketb{\phi}{\psi}= \frac{1}{d} \sum_i \frac{v_{ji}v_{0i}^*}{r_{j \max} r_{0 \max}}= \frac{\braketb{v_0}{v_j}}{d r_{j \max} r_{0 \max}}.
\end{equation}
Note that we do not directly apply the swap test between the two states here because an undesirable contribution to the inner product of the form $\sqrt{1-\frac{r_{jf(j,i)}^2}{r_{j \max}^2}}e^{-i \phi_{jf(j,i)}}\sqrt{1-\frac{r_{0f(0,i)}^2}{r_{0 \max}^2}}e^{i \phi_{0f(0,i)}}$ would arise.  We remove the possibility of such terms appearing by adding an ancilla qubit in~\eq{amppsi} and~\eq{ampphi} that selects out only the component that gives us information about the inner product of $\vec{v}_j$ and $\vec{u}$.

The probability of measuring $0$ in the swap test is $|A|= \frac{1}{2}(1+ |\braketb{\phi}{\psi}|^2)$, which implies that
\begin{equation}
(2|A|-1)d^2r_{j \max}^2r_{0 \max}^2=|\braketb{u}{v_j}|^2.\label{eq:ipsquared}
\end{equation}
At this point we could learn $A$ by sampling from the distribution given by the swap test, but it is more efficient to use amplitude estimation.  Amplitude estimation, in effect, uses a controlled Grover search oracle.  In this case the Grover oracle requires that we implement a reflection about the initial state and also reflect about the space orthogonal to the target state.  We refer to this reflection as $S_\chi$.  Unlike Grover's problem, the reflection about the target state is trivial since the target state is obtained when the swap test yields $0$.  This means that no oracle calls are needed to implement the reflection about the final state.  The reflection about the initial state is of the form $\mathcal{A} S_0 \mathcal{A}^\dagger$, where $\mathcal{A}$ is an algorithm that maps $\ket{0^{2n' +2n+4}}\rightarrow \ket{\phi}\ket{\psi}$, and $S_0$ is of the form
\begin{equation}
S_0\ket{x} = \begin{cases}\ket{x} &, x\ne 0 \\ -\ket{x} & ,x=0 \end{cases}.
\end{equation}
This can be implemented using a multi--controlled $Z$ gate, and hence is efficient.  The prior steps show that $\mathcal{A}$ can be implemented using $6$ oracle calls: \lem{density} implies that three queries are needed for the preparation of $\ket{\psi}$ and three more are needed for the preparation of $\ket{\phi}$.  This implies that a step of Grover's algorithm, given by $\mathcal{A}S_0\mathcal{A}^\dagger S_\chi$, can be implemented using $12$ oracle queries.

Amplitude estimation requires applying an iteration of Grover's algorithm in a controlled fashion at most $R$ times (where $R$ is the Hilbert--space dimension of the register that contains the output of the amplitude estimation algorithm).  The controlled version of $\mathcal{A}S_0\mathcal{A}^\dagger S_\chi$ requires no additional oracle calls because $\mathcal{A}S_0\mathcal{A}^\dagger$ can be made controllable by introducing a controlled version of $S_0$ and using a controlled $R_y$ rotation in the implementation of $\mathcal{A}$ (given by~\lem{density}); furthermore, both $S_0$ and  $S_\chi$ do not require oracle calls and hence can be made controllable at no additional cost. 

The error in the resultant estimate of $P(0)$ after applying amplitude estimation is with probability at least $8/\pi^2$ at most~\cite{BHM+00}
\begin{equation}
|A -\tilde A| \le \frac{\pi}{R}+ \frac{\pi^2}{R^2}.\label{eq:B6}
\end{equation}
This means that $|P(0) -\tilde P(0)|\le \delta$ if
\begin{equation}
R\ge \frac{\pi(\pi+1)}{\delta}.
\end{equation}

Now given that we want the total error, with probability at least $8/\pi^2$, to be $\epsilon/2$ (the factor of $1/2$ is due to the fact that the calculation of $\sin^{-1}(r_{ji}/r_{j \max})$ is inexact) then~\eq{ipsquared} gives us that choosing $\delta$ to be
\begin{equation}
\delta = \frac{\epsilon}{4d^2r_{0 \max}^2r_{j \max}^2},\label{eq:deltabd}
\end{equation}
is sufficient.  Eq.~\eq{deltabd} then gives us that it suffices to take a number of steps of Grover's algorithm that obeys
\begin{equation}
R\ge \left\lceil\frac{4\pi(\pi+1)d^2r_{0 \max}^2r_{j \max}^2}{\epsilon}\right\rceil.\label{eq:rbd}
\end{equation}
Since each Grover iteration requires $12$ applications of the oracles $\mathcal{F}$ or $\mathcal{O}$ the query complexity of the algorithm is
\begin{equation}
Q=12R\le12\left\lceil\frac{4\pi(\pi+1)d^2r_{0 \max}^2r_{j \max}^2}{\epsilon}\right\rceil,\label{eq:B10}
\end{equation}
as claimed.

Now lastly, we need to show that the error in the resultant probabilities from inexactly evaluating $\sin^{-1}(r_{ji}/r_{j \max})$ can be made less than $\epsilon/2$ at polynomial cost.  As shown previously, if $n'\in O(\log(1/\epsilon))$ and if $\ket{\tilde \phi}$ and $\ket{\tilde \psi}$ are the approximate versions of the two inner products then $$\left|\left|\braketb{\phi\Big.}{\psi}\right|^2 - \left|\braketb{\tilde \phi}{\tilde{\psi}}\right|^2\right|\in O(\epsilon).$$ This means that since $\sin^{-1}( \cdot )$ is efficient (it can be computed easily using a Taylor series expansion or approximation by Chebyshev polynomials), the cost of making the numerical error sufficiently small is at most polynomial.
\end{proof}

The proof of \thm{ip} now trivially follows:

\begin{proofof}{Theorem 1}
Proof follows as an immediate consequence of \cor{expamp} and~\lem{ipamp} using $|a_0| \ge 8/\pi^2$ and $r_{j,\max}\le r_{\max}$.
\end{proofof}

\subsection{Proof of Theorem 2}
The structure of the proof of \thm{moment} is similar to that of \thm{ip}.  The biggest difference is that the swap test is not used to compute the Euclidean distance in this case.  We use the method of Childs and Wiebe~\cite{CW12} to perform this state preparation task.  To see how their method works, let us assume that we have access to an oracle $U$ such that for any $j$
\begin{equation}
U\ket{j}\ket{0}= \ket{j}\ket{v_j}.
\end{equation}
Then for any unitary $V\in \mathbb{C}^{\log_2 M \times \log_2 M}$, the following circuit

    \[
      \Qcircuit @R 1em @C 1.5em {
                        \lstick{\ket{0^{\otimes \log_2 M}}} 	&	\gate{V}	&\multigate{1}{U}	&\gate{V^\dagger} 	&\meter\\
		  \lstick{\ket{0^{\otimes n}}}			&	\qw		&\ghost{U}		&\qw				&\qw
},				
    \]
has the property that the probability of the measurement yielding $0$ is $\|\sum_{j=0}^M |V_{j,0}|^2\vec{v_j} \|$.  Thus the Euclidean distance can be computed by this approach by setting $\ket{v_0} = -\ket{u}$ and choosing the unitary $V$ appropriately.
We employ a variant of this in the following lemma, which explicitly shows how to construct the required states.
\begin{lemma}\label{lem:centroid}
For any fixed $\epsilon>0$, the quantities a state of the form $\sqrt{|A|} \ket{\Psi}\ket{y}\ket{y'} +\sqrt{1-|A|} \ket{\Phi_{\rm bad}}$ can be efficiently and reversibly prepared such that $y$ encodes $\|-\ket{u} +\frac{1}{M} \sum_j \ket{v_j}\|_2^2$ within error $\epsilon$ and $y'$ encodes $\frac{1}{M}\sum_p \|-\ket{v_p} +\frac{1}{M} \sum_j \ket{v_j}\|_2^2$ within error $\epsilon$ such that $|A|\ge 64/\pi^4\approx 2/3$ using a number of queries bounded above by
$$
Q\le 10\left\lceil\frac{8\pi(\pi+1)d r_{\max}^2}{\epsilon}\right\rceil,
$$
where $r_{\max}\ge \max_j r_{j \max}$.  
\end{lemma}
\begin{proof}
First let us define an oracle $W$ such that
\begin{equation}
W\ket{j}\ket{p}\ket{i}\ket{0}\ket{0} = \begin{cases} \ket{0}\ket{p}\ket{i}\ket{v_{p,i}}\ket{v_{0,i}}, &j=0\\\ket{j}\ket{p}\ket{i}\ket{v_{ji}}\ket{v_{pi}}, &{\rm otherwise} \end{cases}\label{eq:W}
\end{equation}
A query to $W$ can be implemented via
    \[
      \Qcircuit @R 1em @C 1.5em {
                        \lstick{\ket{j}} 	&	\multigate{2}{O} 	&\qswap 	&\qw		&\multigate{2}{O}	&\qswap 	&\qw			&\ctrlo{2}&\qw\\
		\lstick{\ket{i}}	&	\ghost{O}		&\qw\qwx	&\qw		&\ghost{O}		&\qw\qwx	&\qw			&\qw&\qw\\
		\lstick{\ket{0}}	&	\ghost{O}		&\qw\qwx	&\qswap	&\ghost{O}		&\qw\qwx	&\qswap		&\qswap&\qw\\
		\lstick{\ket{p}}	&	\qw			&\qswap\qwx&\qw\qwx	&\qw					&\qswap\qwx&\qw\qwx		&\qw\qwx&\qw\\
		\lstick{\ket{0}}	&	\qw			&\qw		&\qswap\qwx&\qw					&\qw		&\qswap\qwx&\qswap\qwx&\qw
}				
    \]

We see that the following transformation is efficient and can be performed using one query to $\mathcal{F}$ by applying $V$ to the first register, and applying~\lem{stateprep} to the second and third registers 
\begin{equation}
\ket{0^{\otimes n}}\ket{0^{\otimes m}}\ket{0^{\otimes m}}\ket{0^{\otimes n'}} \ket{0} \mapsto \frac{1}{\sqrt{Md}}\sum_{j=1}^{M}\sum_{p=1}^M \sum_{i=1}^{d} V_{j0} \ket{j}\ket{f(j,i,p)}\ket{p}\ket{0^{\otimes n'}} \ket{0}.
\end{equation}
Here we take $f(j,i,p) = f(j,i)$ if $j\ge 1$ and $f(j,i,p)= f(p,i)$ if $j=0$; furthermore, we use the convention that $\ket{\vec{v}_0}=-\vec{u}$ and $\ket{\vec{v}_0^{(p)}} = - \vec{v}_p$.  We also take $m=\lceil \log_2 M \rceil$ and $n'$ to be the number of bits needed to store the components of each $\vec{v}_j$.

The following state can be implemented efficiently within error at most $\epsilon/2$ using $3$ oracle calls by applying~\lem{density} with the modification that $W$ is used in the place of $O$ and then applying $V^\dagger$:
\begin{equation}
\frac{1}{\sqrt{Md}}\sum_{q,j=0}^M\sum_{i=1}^d \sum_{p=1}^M V_{qj}^\dagger V_{j0}\ket{q}\ket{f(j,i,p)}\ket{p}\left(\sqrt{1-\left(\frac{r_{jf(j,i,p)}^{(p)}}{r_{\max}}\right)^2}e^{-i\phi^{(p)}_{jf(j,i,p)}}\ket{0} + \frac{r^{(p)}_{jf(j,i,p)}}{r_{\max}}e^{i\phi^{(p)}_{jf(j,i,p)}} \ket{1} \right)\ket{\Theta(j,i,p)},\label{eq:state}
\end{equation}
where $\ket{\Theta(j,i,p)}$ is a computational basis state that stores the ancilla qubits prepared by $W$.
The qubit register containing $\ket{\Theta(j,i,p)}$ does not need to be cleaned since these qubits do not affect the trace.

We then use the definition of $V$ in~\lem{V} the probability of measuring the first register in the state $\ket{0^{\otimes n}}$ and the second--last register in the state $\ket{1}$ is:
\begin{align}
{\rm Tr}&\left(\frac{1}{d M r_{\max}^2}\sum_{i,i'} \sum_{j,j'} \sum_{p,p'} V_{j0}V^\dagger_{0j} V_{j'0}^* V^{\dagger *}_{0j'} v_{jp,i} v_{j'p',i}^*\ketbra{i}{i'} \otimes\ketbra{p}{p'} \right)\nonumber\\
&=\left(\frac{1}{d M r_{\max}^2}\sum_{i} \sum_{j,j'} \sum_{p} V_{j0}V^\dagger_{0j} V_{j'0}^* V^{\dagger *}_{0j'} \braketb{i}{v_j^{(p)}}\braketb{v_{j'}^{(p)} }{i} \right)\nonumber\\
&=\left(\frac{1}{Md r_{\max}^2}  \sum_{p} \left| \frac{-1}{2}\ket{v_p} +\frac{1}{2M} \sum_{j\ge 1} \ket{v_j}\right|^2 \right)
\end{align}
We drop the state $\ket{\Theta(j,i,p)}$ above because it does not affect the trace.
Thus the mean square distance between each $\vec{v_j}$ and the centroid is 
\begin{equation}
\frac{1}{M}\sum_p\left| -\ket{v_p} +\frac{1}{M} \sum_{j\ge 1} \ket{v_j}\right|^2=4d r_{\max}^2 P(0).
\end{equation}
Three queries are needed to draw a sample from this distribution.

The distance can be computed similarly, except $O$ can be queried directly instead of $W$ and the ``p''--register can be eliminated since we do not need to average over different distances.  This saves one additional query, and hence it is straightforward to verify that the relationship between the distance squared and the probability of success is
\begin{equation}
\left| -\ket{v_p} +\frac{1}{M} \sum_{j\ge 1} \ket{v_j}\right|^2=4d r_{\max}^2 P(0).
\end{equation}
Two queries are needed to draw a sample from this distribution.

Similar to the proof of~\lem{ipamp}, we use amplitude estimation to estimate $P(0)$ in both cases.  By following the same arguments used in~\eq{B6} to~\eq{B10} coherent AE can be used to prepare a state with probability of the form $\sqrt{|A|}\ket{\Psi_{\rm good}}\ket{y'}+ \sqrt{1-|A|}\ket{\Psi_{\rm bad};y'_\perp}$ where $|A|\le 8/\pi^2$ using a number of queries bounded above by
\begin{equation}
6\left\lceil\frac{8\pi(\pi+1) d r_{\rm max}^2}{\epsilon}\right\rceil.\label{eq:cost}
\end{equation}
Similarly, the cost of preparing a state of the form  $\sqrt{|A|}\ket{\Psi}\ket{y}+ \sqrt{1-|A|}\ket{\Phi;y_\perp}$ where $|A|\le 8/\pi^2$ and $\ket{y}$ encodes $|\vec{u}-{\rm mean}(\{\vec{v}_j\})|$ is
\begin{equation}
4\left\lceil\frac{8\pi(\pi+1) d r_{\rm max}^2}{\epsilon}\right\rceil.\label{eq:cost2}
\end{equation}
Therefore, by combining~\eq{cost} and~\eq{cost2} we see that a state of the form $\sqrt{|A|}\ket{\Psi}\ket{y}\ket{y'}+ \sqrt{1-|A|}\ket{\Phi_{\rm bad}}$ can be constructed where $|A|\le (8/\pi^2)^2\approx 2/3$ using a number of oracle calls bounded above by
\begin{equation}
10\left\lceil\frac{8\pi(\pi+1) d r_{\rm max}^2}{\epsilon}\right\rceil,\label{eq:cost2b}
\end{equation}
\end{proof}

\thm{moment} follows trivially from the above results and the proof of this result is given below.
\begin{proofof}{Theorem 2}
Proof follows as a trivial consequence of taking $M=M'$ in~\cor{expamp}, applying~\lem{centroid} and observing that dividing the calculated distance by $\sigma_m$ is efficient irrespective of whether $M'>1$ or $M'=1$.  Note that the upper bound is not tight in cases where $M'=1$ because $\sigma_m$ does not need to be computed in such cases; nonetheless, removing this cost only reduces the expected query complexity by a constant factor so we do not change the theorem statement for simplicity.  Also, using a non--constant value for $M_1,\ldots,M_{M'}$ does not change the problem since the state preparation method of \lem{stateprep} takes $M$ as an input state that can be set to $M_1,\ldots,M_{M'}$ coherently.
\end{proofof}

With these results in hand we can now prove~\cor{moment}, which gives the query complexity of performing an iteration of $k$--means.
\begin{proofof}{\cor{moment}}
The first fact to notice is that unlike~\thm{moment}, the algorithm for $k$--means does not require that we normalize the distance.  This reduces the cost of the algorithm by a factor of $4/10$ using~\eq{cost} and~\eq{cost2}.  Also because the algorithm succeeds with probability $\pi/8$ rather than $(\pi/8)^2$, the costs are further reduced by a factor of $((\pi/8)^2-1/2)^2/((\pi/8)-1/2)^2$.
Thus we see from these observations that the total cost of the algorithm is 
$$360M\sqrt{k} \left\lceil\frac{8\pi(\pi+1) d r_{\max}^2}{\epsilon} \right\rceil\left\lceil\frac{\log\left(\frac{81k(\log(k) +\gamma)}{\delta_0} \right)}{2((8/\pi^2) -1/2)^2} \right\rceil.$$
\end{proofof}

\section{Justification for normalizing distance}\label{app:ratio}
An important question remains: when is the normalized distance between $\vec{u}$ and the cluster centroid a useful statistic in a machine learning task?  
Of course, as mentioned earlier, this statistic is not always optimal.
In cases where the training data points live on a complicated manifold, there will be many points that are close to the cluster centroid yet are not in the cluster.  
Even in such circumstances, the normalized distance leads to an upper bound on the probability that $\vec{u}$ is in the cluster.

For concreteness, let us assume that 
\begin{align}
|\vec{u} - {\rm mean}(\{A\})| &= \xi_A,\nonumber\\
|\vec{u} - {\rm mean}(\{B\})| &= \xi_B.
\end{align}
If we then define the intra--cluster variances to be $\sigma_A^2$ and $\sigma_B^2$ for clusters $\{A\}$ and $\{B\}$, then Chebyshev's inequality states that regardless of the underlying distributions of the clusters that for any point $x$
\begin{align}
\Pr(|x-{\rm mean}(\{A\})| \ge \xi_A| x\in \{A\}) &\le \frac{\sigma_A^2}{\xi_A^2},\nonumber\\
\Pr(|x-{\rm mean}(\{B\})| \ge \xi_B| x\in \{B\}) &\le \frac{\sigma_B^2}{\xi_B^2}.\label{eq:chebyshev}
\end{align}
Eq.~\eq{chebyshev} tells us that if the normalized distance is large then the probability that the point is in the corresponding cluster is small.

Unfortunately,~\eq{chebyshev} does not necessarily provide us with enough information to merit use in a decision problem because there is no guarantee that the inequalities are tight.  If Chebyshev's inequality is tight then basing a decision on the normalized distance is equivalent to the likelihood--ratio test, which is widely used in hypothesis testing.

\begin{theorem}\label{thm:justify}
Assume there exist positive numbers $a,b,\alpha,\beta$ such that for all $\chi\ge \min\{\xi_A,\xi_B \}$
\begin{align*}
a \frac{\sigma_A^2}{\chi^2}\le& \Pr(|x-{\rm mean}(\{A\})| \ge \chi|x\in \{A\})\le \alpha \frac{\sigma_A^2}{\chi^2}\\
b \frac{\sigma_B^2}{\chi^2}\le& \Pr(|x-{\rm mean}(\{B\})| \ge \chi|x\in \{B\})\le \beta \frac{\sigma_B^2}{\chi^2},
\end{align*}
and either $a\ge \beta$ or $\alpha \le b$, then using the normalized distance to the cluster centroid to decide whether $\vec{u}\in \{A\}$ or $\vec{u}\in \{B\}$ is equivalent to using the likelihood ratio test.
\end{theorem}
\begin{proof}
The likelihood ratio test concludes that $\vec{u}$ should be assigned to $A$ if 
\begin{equation}
\frac{\Pr(|\vec{u}-{\rm mean}(\{A\})| \ge \chi|x\in \{A\})}{ \Pr(|\vec{u}-{\rm mean}(\{B\})| \ge \chi|x\in \{B\})}> 1.\label{eq:prbound}
\end{equation}
Our assumptions show that~\eq{prbound} is implied if
\begin{equation}
\frac{a \frac{\sigma_A^2}{\xi_A^2}}{\beta \frac{\sigma_B^2}{\xi_B^2}}> 1\Rightarrow \frac{a}{\beta}> \left(\left(\frac{\sigma_B}{\xi_B}\right)\left(\frac{\sigma_A}{\xi_A}\right)^{-1} \right)^2.\label{eq:likelihoodA}
\end{equation}
If the normalized distance is used as the classification decision, then $\vec{u}$ is assigned to $\{A\}$ if $\left(\left(\frac{\sigma_B}{\xi_B}\right)\left(\frac{\sigma_A}{\xi_A}\right)^{-1} \right)\le 1$.  Therefore the two tests make the same assignment if $a\ge\beta$.

The likelihood ratio test similarly assigns $\vec{u}$ to $\{B\}$ if 
\begin{equation}
\frac{\Pr(|\vec{u}-{\rm mean}(\{A\})| \ge \chi|x\in \{A\})}{ \Pr(|\vec{u}-{\rm mean}(\{B\})| \ge \chi|x\in \{B\})}< 1,
\end{equation}
which, similar to~\eq{likelihoodA} is implied by $\frac{\alpha}{b}< \left(\left(\frac{\sigma_B}{\xi_B}\right)\left(\frac{\sigma_A}{\xi_A}\right)^{-1} \right)^2$ and is further equivalent to the distance--based assignment if $\alpha \le b$.
\end{proof}

\thm{justify} shows that the validity of distance--based assignment depends strongly on the tightness of Chebyshev's bound; however, it is not necessarily clear a priori whether lower bounds on $\Pr(|\vec{u}-{\rm mean}(\{A\})| \ge \chi|x\in \{A\})$ and $\Pr(|\vec{u}-{\rm mean}(\{B\})| \ge \chi|x\in \{B\})$ exist for values of $a$ and $b$ that are non--zero.  Such bounds clearly exist if, for example, $\{A\}$ and $\{B\}$ are both drawn from Gaussian distributions. This follows because $\Pr(|\vec{u}-{\rm mean}(\{B\})| \ge \chi|x\in \{B\})$ can be upper-- and lower--bounded by a function of the distance and appropriate values of $a$ and $\alpha$ can be extracted from the covariance matrix.  Since both distributions are the same in this case, $a=b$, and $\alpha=\beta$, the normalized distance is well motivated if $\{A\}$ and $\{B\}$ are drawn from two Gaussian distributions whose centroids are sufficiently distant from each other.

\section{Sensitivity of decision problem}\label{app:conc}
Although our quantum algorithms for computing the inner product and Euclidean distance provide better scaling with $N$ and $M$ (the dimension of the vectors and the number of vectors in the training set) than their classical analogs, the quantum algorithms introduce a $O(1/\epsilon)$ scaling with the noise tolerance (where $\epsilon$ is the error tolerance in the distance computation).  If the typical distances in the assignment set shrink as $1/N^{\gamma}$ for positive integer $\gamma$, then it is possible that the savings provided by using a quantum computer could be negated because $\epsilon^{-1}\in \Omega(N^\gamma)$ in such cases.

We will now show that in ``typical'' cases where the vectors are uniformly distributed over the unit sphere that $\epsilon \in \Theta(1/\sqrt{N})$ will suffice with high probability.  Since our nearest--neighbor algorithms scale as $\sqrt{M}/\epsilon$ (which in the Euclidean case corresponds to $M'=M$) this implies that our algorithms' cost scales as $O(\sqrt{NM})$. 
This scales quadratically better than its classical analog or Monte--Carlo sampling.  

This result is similar to concentration of measure arguments over the hypersphere, which show that almost all unit vectors are concentrated in a band of width $O(1/\sqrt{N})$ about any equator of the hypersphere \cite{Led05}.  The concentrated band of vectors is the origin of the ``curse of dimensionality'', which implies that almost all random unit vectors are within a Euclidean distance of $\sqrt{2} - O(1/\sqrt{N})$ of any fixed unit vector.  This means that the underlying distribution of distances in nearest--neighbor learning tends to flatten as $N\rightarrow \infty$, implying that very accurate estimates of the distances may be needed in high--dimensional spaces.  

For Haar--random vectors, it can be shown that the probability distribution for the magnitude of each component of the vector is (to within negligible error in the limit of large $N$) independently distributed and has a probability density of~\cite{PR03}
\begin{equation}
\rho(|v_{jk}|=r)= 2(N-1)r(1-r^2)^{N-2},
\end{equation}
which after a substitution gives
\begin{equation}
\rho(\sqrt{N}|v_{jk}|=u)= \frac{2(N-1)u}{\sqrt{N}}\left(1-\frac{u^2}{N}\right)^{N-2}\sim 2u\sqrt{N} e^{-u^2}=2rNe^{-Nr^2}.\label{eq:r}
\end{equation}
Chebyshev's inequality then can be used to show that with high probability $r\in \Theta(1/\sqrt{N})$ for each component, and~\eq{r} shows that the distribution varies smoothly and hence it is highly probable that the differences between any two components of such random vectors are $\Theta(1/\sqrt{N})$.

Since the Haar measure is invariant under unitary transformation, we can take $\vec{u}$ to be a basis vector without loss of generality.  Then if we define $\vec{w}$ and $\vec{z}$ to be the two closest vectors we see that with high probability in the limit as $N\rightarrow \infty$
\begin{align}
|\vec{u} -\vec{w}|^2 &= (1-\vec{w}_1)^2 + \sum_{j=2}^N \vec{w}_j^2 \nonumber\\
&= (1-\vec{w}_1)^2 + (1-|\vec{w}_1|^2)\nonumber\\
&= 1 +  (1-\vec{w}_1)^2 + O(1/N)\nonumber\\
& = 2 + \Theta(1/\sqrt{N}),
\end{align}
where the last line follows from the observation that with high probability  $|\vec{w}_1|\in \Theta(1/\sqrt{N})$.  By repeating the same argument we see that
\begin{equation}
|\vec{u} -\vec{z}|^2  = 2 +\Theta(1/\sqrt{N}),
\end{equation}
and hence from the fact that the distribution of distances is smooth and the components of $\vec{w}$ and $\vec{z}$ are independent we see that
\begin{equation}
|\vec{u} -\vec{z}|^2-|\vec{u} -\vec{w}|^2  \in \Theta(1/\sqrt{N}).
\end{equation}
This suggests that $\epsilon\in O( 1/\sqrt{N})$ for the case where the members of the training set are Haar--random vectors.  We demonstrate this point numerically below.


\begin{figure}[h!]
\centering
\includegraphics[width=0.5\linewidth]{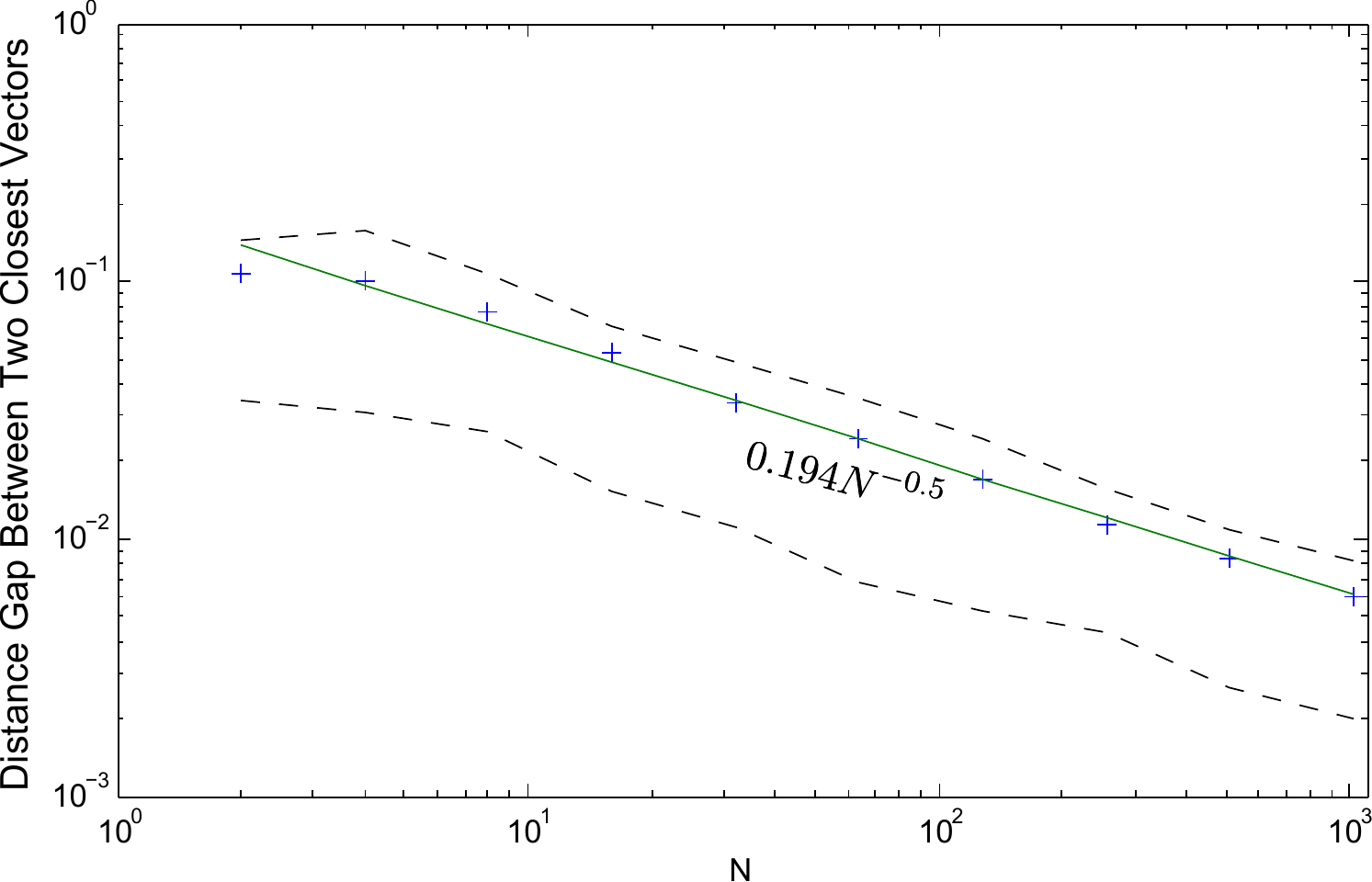}
\caption{\label{fig:sensitivity} Mean difference in Euclidean distance between two closest vectors to $\vec{u}$ for vectors randomly chosen according to the Haar measure.  The plot is computed for $M=100$ using $100$ trials per value of $N$.  The blue crosses show the mean values of the distances and the dashed lines give a $50\%$ confidence interval on the distance.  The green line gives the the best powerlaw fit to the data.}
\end{figure}

We generate the data in~\fig{sensitivity} by generating a large set of random vectors chosen uniformly with respect to the Haar measure.  We take $\ket{u}=\ket{0}$ in all these examples without loss of generality because the measure is rotationally invariant.  We then compute for each $j=1,\ldots,M$ $|\vec{u}-\vec{w}_j|_2$ and sort the distances between $\vec{u}$ and each of the randomly drawn vectors.  Finally, we compute the distance gap, or the difference between the two smallest distances, and repeat this $100$ times in order to get reliable statistics about these differences.  

\fig{sensitivity} shows that the difference between these two distances tends to be on the order of $1/\sqrt{N}$ as anticipated from concentration of measure arguments.  It is easy to see from Taylor's theorem that the differences in the square distances is also $O(1/\sqrt{N})$. 
Hence taking $\epsilon\in O(1/\sqrt{N})$ will suffice with high--probability since an error of this scale or smaller will not be sufficient to affect the decision about the identity of the nearest vector.

In contrast, the scaling with $M$ is much less interesting because the volume expands exponentially with $N$.
Hence it takes a very large number of points to densely cover a hypersphere.  For this reason, we focus on the scaling with $N$ rather than $M$.
However, for problems with small $N$, large $M$, and no discernable boundaries between $U$ and $V$, this issue could potentially be problematic.

We can now estimate the regime in which our quantum algorithms will have a definite advantage over a brute--force classical computation.  We assume that $\epsilon=1/\sqrt{N}$, $\delta_0=0.5$, $dr_{\max}^2=1$ and $M'=M$.  We then numerically compute the points where the upper bounds on the query complexity in \thm{ip} and \thm{moment} equal the cost of a brute--force classical computation.  We use these points to estimate the regime where our quantum algorithms be cost--advantageous over classical brute--force classification.  
As seen in \fig{costregion}, our quantum algorithms exhibit superior time complexities for a large range of $M$ and $N$ values.
This trade--off point for the Euclidean method occurs when $M \approx 10^{16} N^{-1.07}$, and $M\approx 2\times 10^{14} N^{-1.08}$ for the inner--product method.

It is important to note that the upper bounds on the query complexity are not expected to be tight, which means that we cannot say with confidence that our quantum algorithms will not be beneficial if the upper bounds are less than $NM$.  
Tighter bounds on the query complexity of the algorithm may be needed in order to give a better estimate of the performance of our algorithm in typical applications.

\begin{figure}[h!]
\centering
\includegraphics[width=0.5\linewidth]{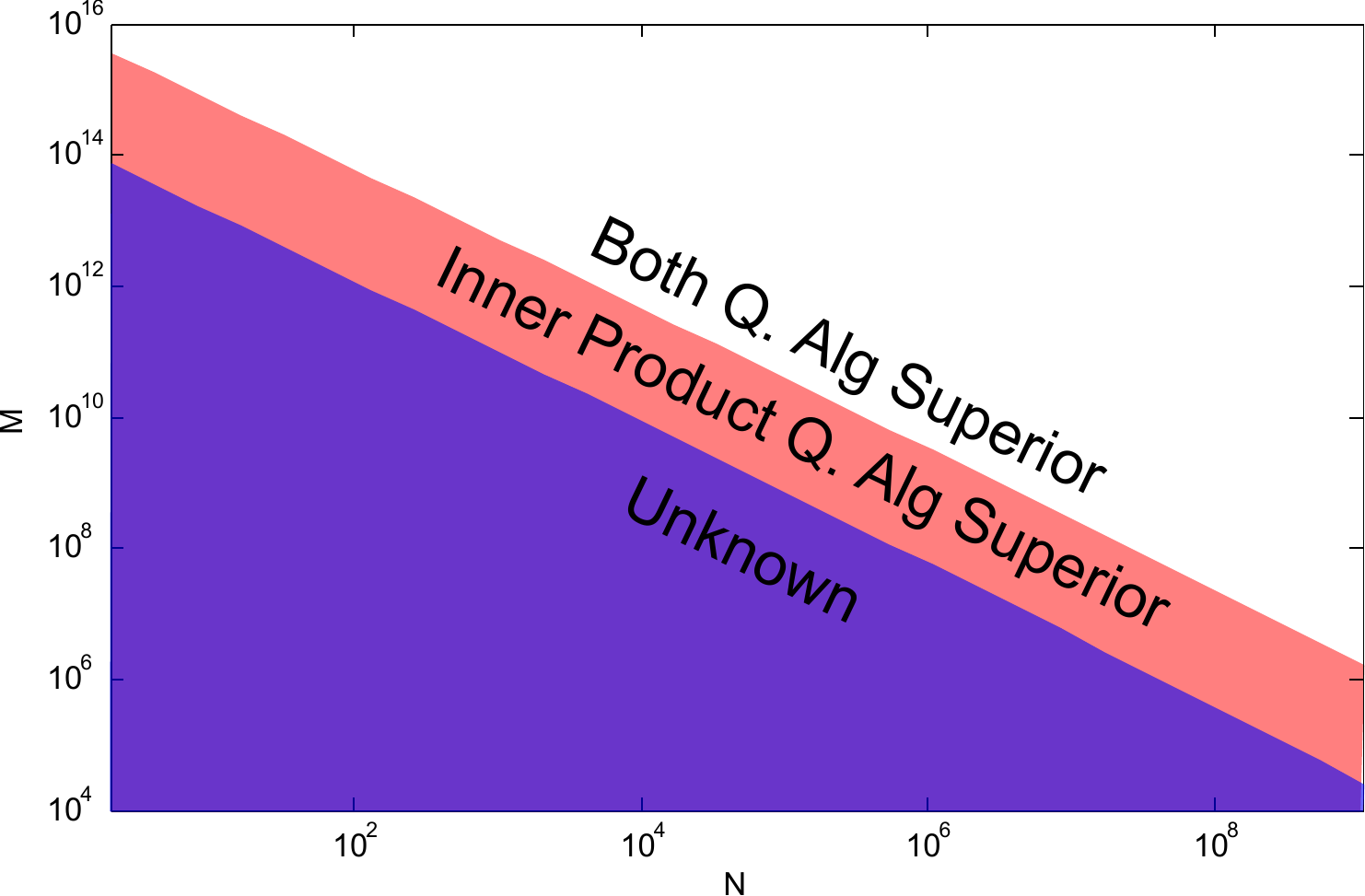}
\caption{\label{fig:costregion} Estimated regions where our quantum algorithms are cost-advantageous over a brute--force classical calculation.  The shaded regions represent the parameter space where the upper bounds from \thm{ip} and \thm{moment} are greater than the brute--force classical cost of $NM$. }
\end{figure}

\section{$k$--nearest--neighbor classification}\label{app:kNN}
In practice, nearest--neighbor classification runs into difficulties when faced with errors in the labels in the training set or when faced with data sets that have a significant number of outliers.  One way to combat these issues is to use $k$--nearest--neighbor classification where the class label is chosen to be the mode of the labels of the $k$ closest training vectors to the test vector.  This makes the assignment much more robust to both errors by averaging over the labels of the closest vectors to the test vector.

The cost of generalizing a nearest--neighbor classification algorithm to a $k$--nearest--neighbor algorithm is given below.
\begin{lemma}
Assume that the query complexity of classifying a test vector using either the nearest--neighbor Euclidean method or the inner--product method is $Q$, then the query complexity of performing either algorithm as a $k$--nearest--neighbor algorithm is at most $kQ$.
\end{lemma}
\begin{proof}
Our proof proceeds inductively.  The claim is clear for the case where $k=1$, which demonstrates the base case.  Now assume that we know the $k'$ nearest--neighbors to the test vector and wish to find the $k'+1$ vector.  This can be done by applying the same steps used in either classifier after removing the first $k'$ vectors from the set.  The complexity of performing D\"urr H\o yer optimization over the distances computed for this reduced set of vectors is given by~\lem{durr} to be at most $$\frac{45}{2} \sqrt{M-k'}\le \frac{45}{2}\sqrt{M}.$$
This implies that the cost of doing so is at most $Q$ because the only step of the algorithm that needs to change is the search at the outermost loop of either of our algorithms, hence the total cost is at most $(k'+1)Q$ queries.
Thus if we can remove these vectors from consideration then it follows by induction that the cost of nearest--neighbor classification is at most $kQ$ queries.

The vectors do not need to be directly removed from consideration.  Rather, the Grover oracle used in the D\"urr H\o yer algorithm can be modified to ignore the $k'$ closest vectors from consideration.  Let $\{{\rm marked}\}$ denote the marked set for the Grover oracle, then it is straightforward to construct a circuit that marks the set $\{\rm marked\}\setminus \{\vec{v}_{1},\ldots,\vec{v}_{k'}\}$ where we assume without loss of generality that $\vec{v}_1,\ldots,\vec{v}_{k'}$ are the $k'$ closest vectors.  Since these vectors are known classically, we do not need any further queries to remove these vectors from consideration in the modified Grover oracle.  Thus we can assume that the first $k'$ vectors can be effectively removed from the training set without changing the number of queries, which completes our proof.
\end{proof}

Other algorithms for computing the $k$ nearest--neighbors exist.  The neighbors can be individually estimated using the methods of~\cite{NW99}, which provides an $O(\sqrt{k}\log(k){\rm loglog}(k))$ algorithm for approximating the smallest element.  Although not directly more efficient, since all $k$ such values must be found, it is better suited for parallel quantum computation because each of the $k$ nearest--neighbors can be found independently. Another heuristic approach would be to use the D\"urr H\o yer algorithm to focus on a subset of points and then sample from those points to find the $k$--nearest vectors.

\section{Pseudocode}\label{app:pseudo}
For clarity, we provide pseudocode explaining how to perform several of the more important steps in the quantum algorithm including coherent amplitude estimation and calculation of distances using the inner product and Euclidean methods.  In particular, we explicitly show how to compute distances explicitly using the inner product method for non--positive vectors.
\begin{figure}[h!]
\includegraphics[width=0.9\linewidth]{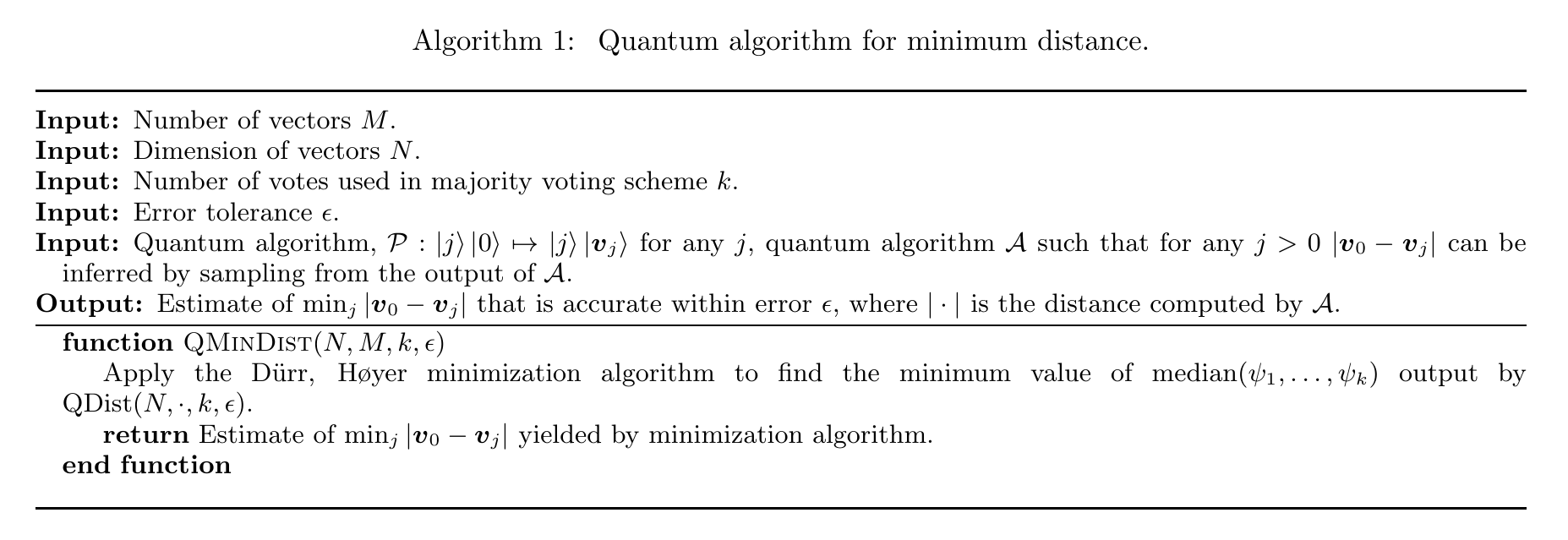}
\end{figure}
\begin{figure}[h!]
\includegraphics[width=0.9\linewidth]{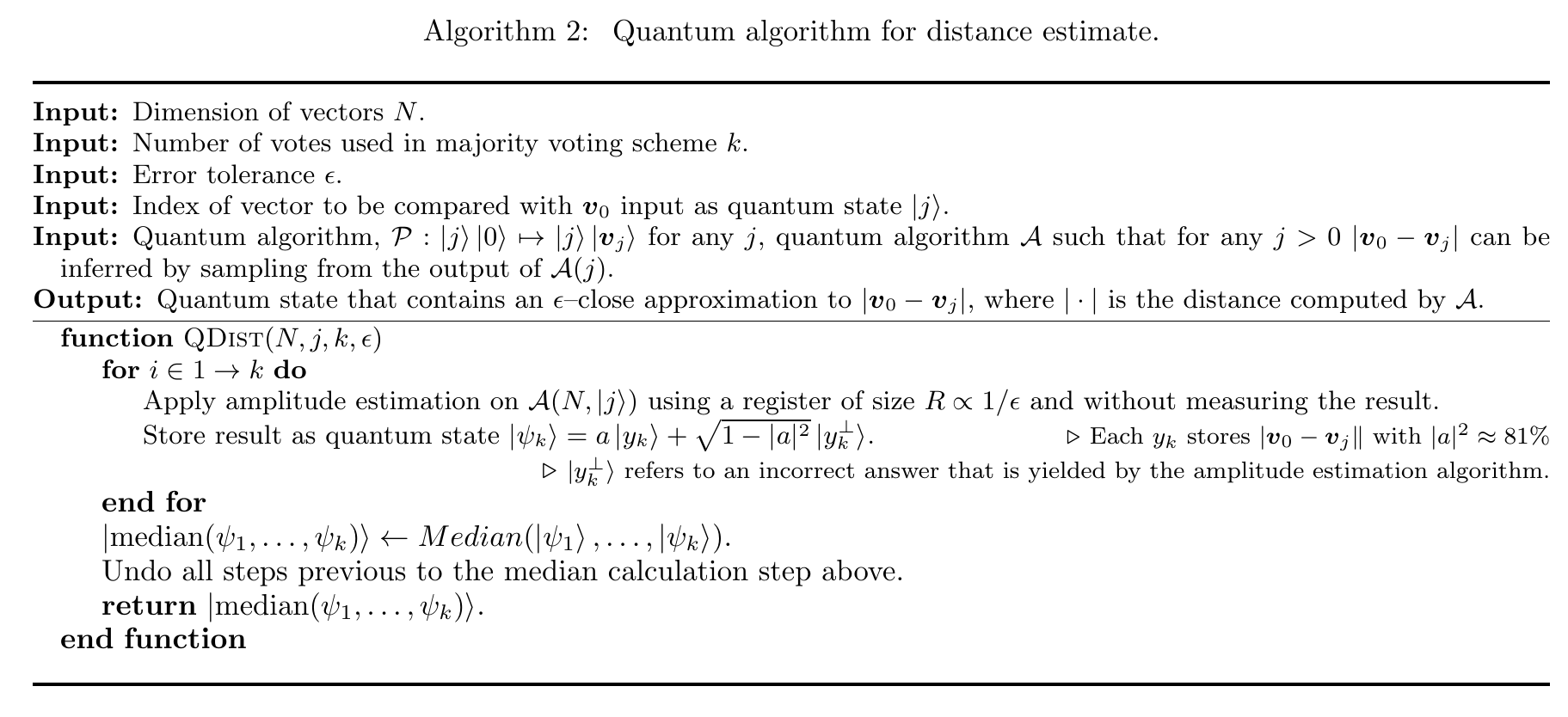}
\end{figure}
\begin{figure}[h!]
\includegraphics[width=0.9\linewidth]{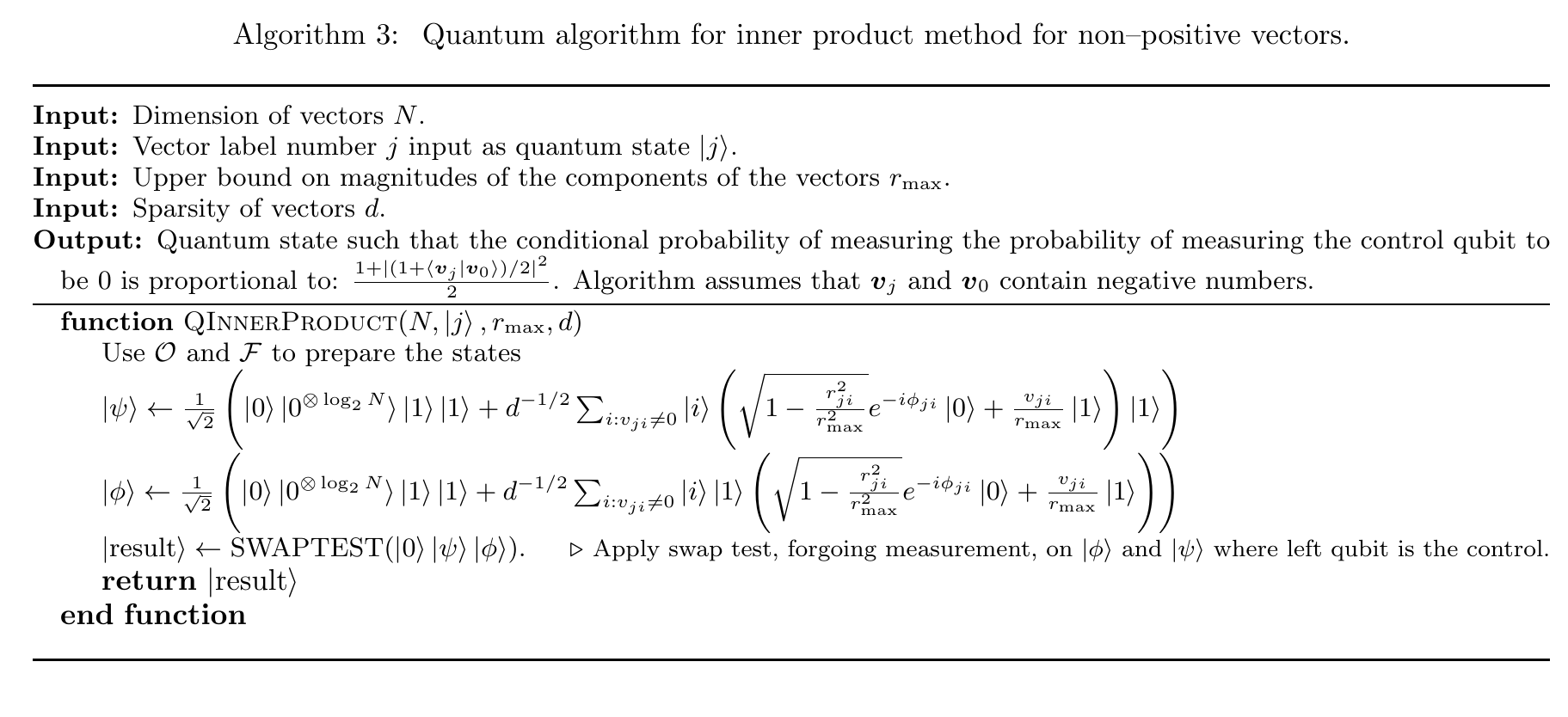}
\end{figure}
\begin{figure}[h!]
\includegraphics[width=0.9\linewidth]{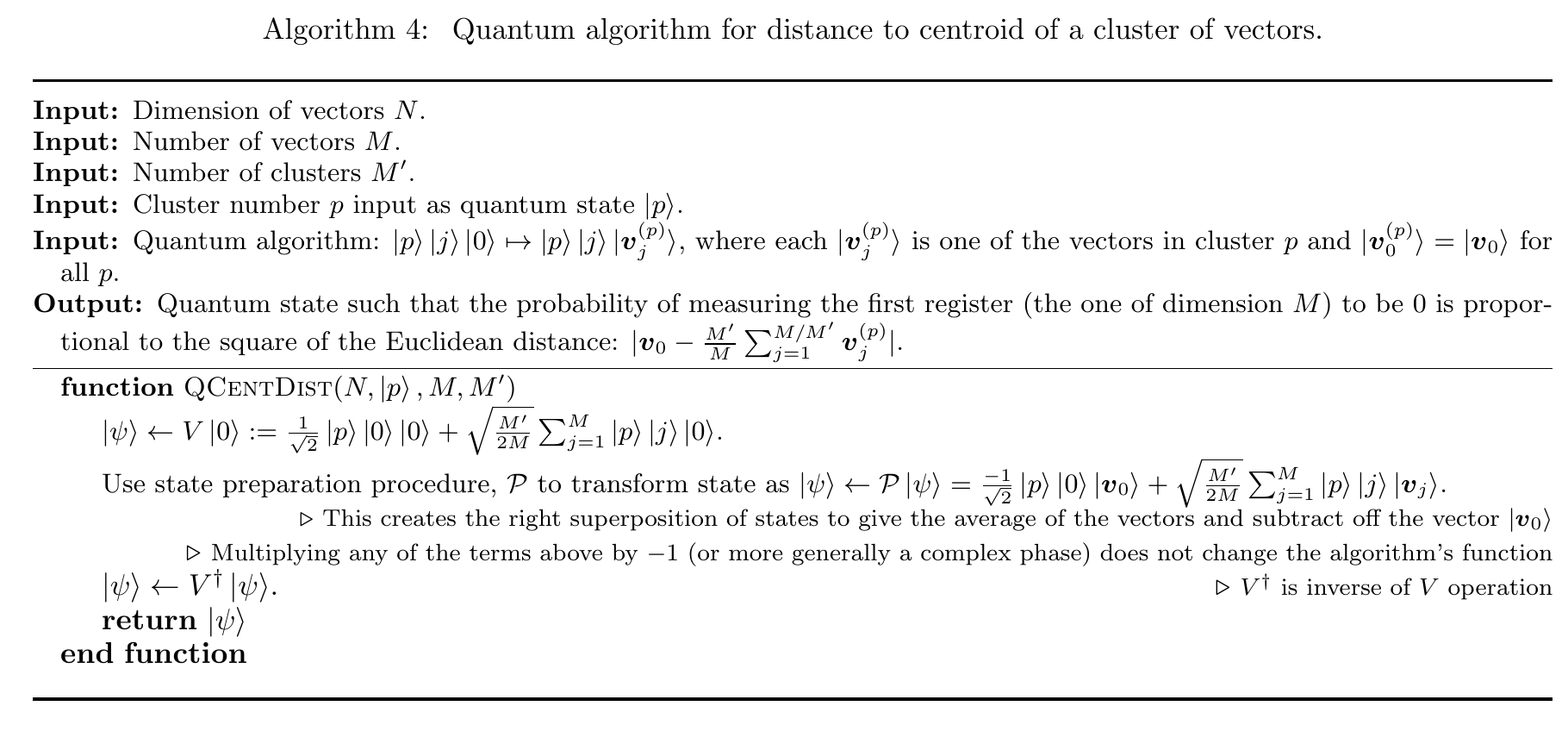}
\end{figure}
\begin{figure}[h!]
\includegraphics[width=0.9\linewidth]{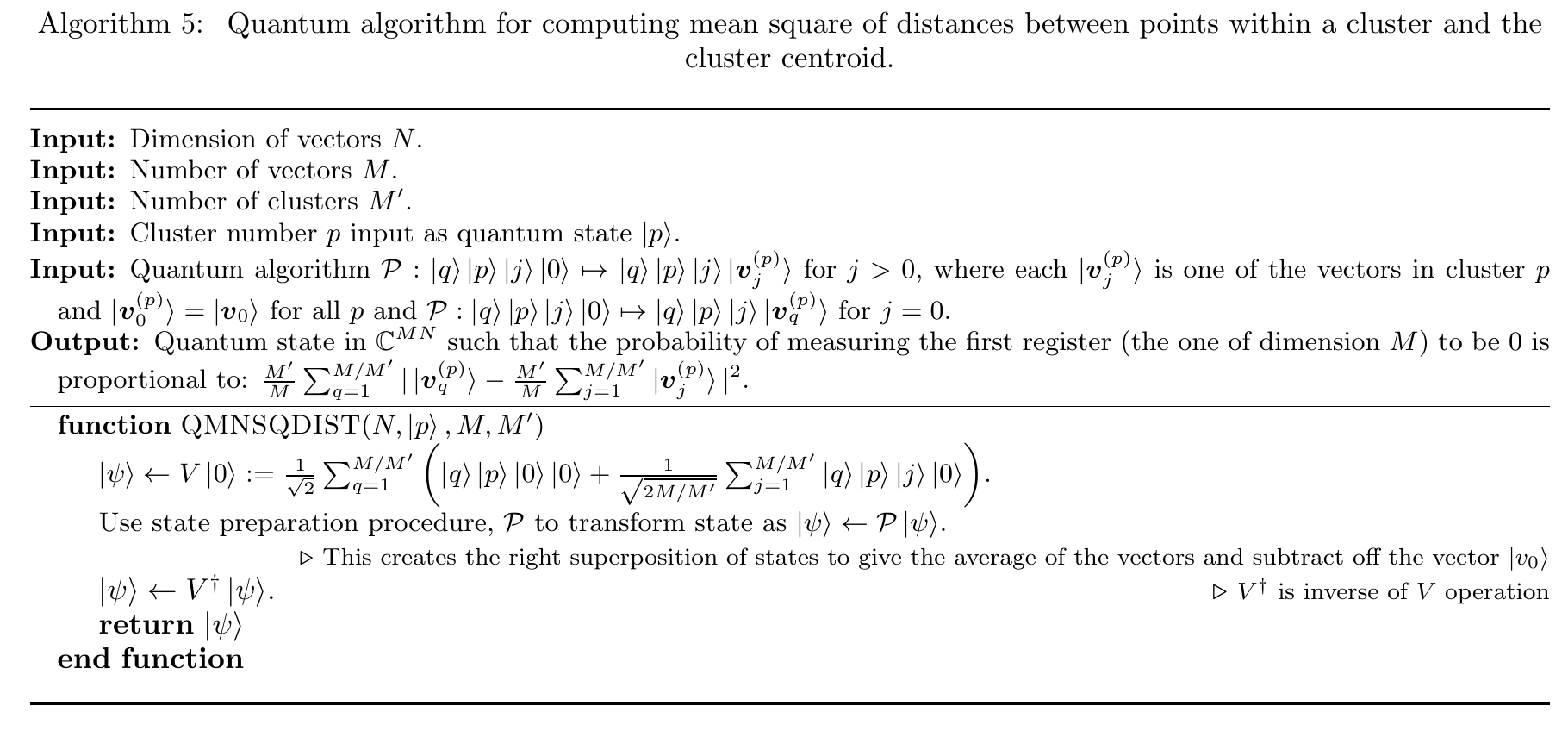}
\end{figure}
\end{document}